\documentclass[10pt,twocolumn,twoside]{IEEEtran}

\normalsize
\usepackage{amsmath}
\usepackage{color}
\usepackage{cases}
\usepackage{amsfonts,amsmath,amssymb}
\usepackage{mathrsfs}
\usepackage{cite}
\usepackage{graphicx}
\usepackage{url}
\usepackage{enumerate}
\usepackage{subfigure}
\usepackage{hyperref}
\usepackage{cite}
\usepackage{float}
\usepackage{filecontents}
\usepackage{pst-bar}
\newtheorem{theorem}{Theorem}
\newtheorem{proposition}{Proposition}
\newtheorem{lemma}{Lemma}
\newtheorem{corollary}{Corollary}
\newtheorem{example}{Example}
\newtheorem{remark}{Remark}
\newtheorem{definition}{Definition}

\ifCLASSINFOpdf
\else
\fi

\begin{document}
\title{Error Correction in Polynomial Remainder Codes with Non-Pairwise Coprime Moduli and Robust Chinese Remainder Theorem for Polynomials}

\author{Li Xiao \,\,\,\, and  \,\,\,\,  Xiang-Gen Xia

\thanks{The authors are with Department of
Electrical and Computer Engineering,
University of Delaware, Newark, DE 19716, U.S.A. (e-mail:
\{lixiao, xxia\}@ee.udel.edu). Their work was supported in part
by the Air Force Office of Scientific Research (AFOSR) under
Grant FA9550-12-1-0055.}
}
\maketitle

\begin{abstract}
This paper investigates polynomial remainder codes with non-pairwise coprime moduli. We first consider a robust reconstruction problem for polynomials from erroneous residues when the degrees of all residue errors are assumed small, namely robust Chinese Remainder Theorem (CRT) for polynomials. It basically says that a polynomial can be reconstructed from erroneous residues such that the degree of the reconstruction error is upper bounded by $\tau$ whenever the degrees of all residue errors are upper bounded by $\tau$, where a sufficient condition for $\tau$ and a reconstruction algorithm are obtained. By releasing the constraint that all residue errors have small degrees, another robust reconstruction is then presented when there are multiple unrestricted errors and an arbitrary number of errors with small degrees in the residues. By making full use of redundancy in moduli, we obtain a stronger residue error correction capability in the sense that apart from the number of errors that can be corrected in the previous existing result, some errors with small degrees can be also corrected in the residues. With this newly obtained result, improvements in uncorrected error probability and burst error correction capability in a data transmission are illustrated.
\end{abstract}


\IEEEpeerreviewmaketitle

\section{Introduction}

\IEEEPARstart{T}{he} Chinese Remainder Theorem (CRT) can uniquely determine a large integer from its remainders with respect to several moduli if the large integer is less than the least common multiple (lcm) of all the moduli \cite{CRT1,CRT2}. Based on the CRT for integers, residue codes with pairwise or non-pairwise coprime moduli are independently constructed, where codewords are residue vectors of integers in a certain range modulo the moduli. More specifically, a residue code with pairwise coprime moduli $m_1,\cdots,m_k,m_{k+1},\cdots,m_n$ consists of residue vectors of integers in the range $[0,\prod_{i=1}^{k}m_i)$, where the first $k$ moduli form a set of nonredundant moduli, and the last $n-k$ moduli form a set of redundant moduli used for residue error detection and correction. Over the past few decades, there has been a vast amount of research on residue error correction algorithms for such a class of codes. For more details, we refer the reader to \cite{kls1,kls2,vgoh,grs,rsk,e121,e122}. By removing the requirement that the moduli be pairwise coprime, a residue code with non-pairwise coprime moduli $m_1,m_2,\cdots,m_l$ consists of residue vectors of integers in the range $[0,\mbox{lcm}(m_1,m_2,\cdots,m_l))$. Compared with residue codes with pairwise coprime moduli, the residue error detection and correction algorithm for residue codes with non-pairwise coprime moduli is much simpler, and the price paid for that is an increase in redundancy. Moreover, residue codes with non-pairwise coprime moduli may be quite effective in providing a wild coverage of ``random" errors \cite{bm,bm1,bm2}. In order to perform reliably polynomial-type operations (e.g., cyclic convolution, correlation, DFT and FFT computations) with reduced complexity in digital signal processing systems, residue codes over polynomials (called polynomial remainder codes in this paper) with pairwise or non-pairwise coprime polynomial moduli have been investigated as well \cite{poly5,poly6,poly8,poly1,poly2,poly3,peb,ss}, where codewords are residue vectors of polynomials with degrees in a certain range modulo the moduli and all polynomials are defined over a Galois field. Polynomial remainder codes are a large class of codes that include BCH codes and Reed-Solomon codes as special cases \cite{poly4,poly44}. Due to two important features in residue codes: carry-free arithmetics and absence of ordered significance among the residues, residue error detection and correction technique in residue codes has various applications in, for example, fault-tolerant execution of arithmetic operations in digital processors and in general digital hardware implementations on computers \cite{mhe,why,eddc,peb,ss}, orthogonal frequency division multiplexing (OFDM) and  code division multiple access (CDMA) based communication systems \cite{lly1,lly2,lly3,lle,lle2,ass1,papr}, and secure distributed data storage for wireless networks \cite{dad1,dad2,dad3,dad4}.

In this paper, we focus on polynomial remainder codes with non-pairwise coprime moduli. Note that a coding theoretic framework for such a class of codes has been proposed in \cite{ss}, where the concepts of Hamming weight, Hamming distance, code distance in polynomial remainder codes are introduced. It is stated in \cite{ss} that a polynomial remainder code with non-pairwise coprime moduli and code distance $d$ can correct up to $\lfloor(d-1)/2\rfloor$ errors in the residues, and a fast residue error correction algorithm is also presented, where $\lfloor\star\rfloor$ is the floor function. This reconstruction from the error correction method is accurate but only a few of residues are allowed to have errors and most of residues have to be error-free. The goal of this paper is to study robust reconstruction and error correction when a few residues have arbitrary errors (called unrestricted errors) similar to \cite{ss} and some (or all) of the remaining residues have small errors (i.e., the degrees of errors are small). It is two-fold. One is to study  robust reconstruction and the other is to study error correction, i.e., accurate reconstruction, when residues have errors.

Considering instabilities of data processing in wireless sensor networks and signal processing systems, robust reconstructions based on the CRT for integers were recently studied in \cite{chessa,wenjie,xiao} with different approaches. In this paper, by following the method in \cite{xiao} together with the error correction algorithm for polynomial remainder codes in \cite{ss}, we first propose a robust reconstruction algorithm for polynomials from erroneous residues, called robust CRT for polynomials, i.e., a polynomial can be reconstructed from erroneous residues such that the degree of the reconstruction error is upper bounded by the robustness bound $\tau$ whenever the degrees of all residue errors are upper bounded by $\tau$, where a sufficient condition for $\tau$ for the robustness to hold is obtained.
Next, by releasing the constraint that the degrees of all residue errors have to be bounded by $\tau$, we propose another robust reconstruction algorithm when a combined occurrence of multiple unrestricted errors and an arbitrary number of errors with degrees upper bounded by $\lambda$ happens to the residues, where a sufficient condition for $\lambda$ is also presented in this paper. Note that a combined occurrence of a single unrestricted error and an arbitrary number of small errors in the residues was considered for the robust reconstruction based on the CRT for integers in \cite{chessa}, but its approach is hard to deal with the case of multiple unrestricted errors combined with small errors in the residues due to a considerable decoding complexity. A detailed comparison in terms of robust reconstruction between this paper and \cite{chessa,xiao} is pointed out later in this paper (see Remark \ref{houm}). One can see that the above reconstructions may not be accurate but robust to the residue errors in terms of degree and all the residues are allowed to have errors.

Finally, we consider the residue error correction in a polynomial remainder code with code distance $d$. Compared with the result in \cite{ss}, by making full use of the redundancy in moduli and newly proposed robust reconstruction method, we obtain a stronger residue error correction capability in the sense that apart from correcting up to $\lfloor(d-1)/2\rfloor$ residue errors, a polynomial remainder code with code distance $d$ can correct some additional residue errors with small degrees. With this newly obtained result, improvements in the performances of uncorrected error probability and burst error correction considered in a data transmission are illustrated.

The rest of the paper is organized as follows. In Section \ref{sec1}, we briefly introduce some fundamental knowledge in polynomials over a Galois field and coding theory of polynomial remainder codes with non-pairwise coprime moduli obtained in \cite{ss}. In Section \ref{sec2}, we propose robust CRT for polynomials. In Section \ref{sec4}, another robust reconstruction is considered when a combined occurrence of multiple unrestricted errors and an arbitrary number of errors with small degrees is in the residues. In Section \ref{sec3}, a stronger residue error correction capability in polynomial remainder codes with non-pairwise coprime moduli and its improvements in uncorrected error probability and burst error correction in a data transmission are presented. We conclude this paper in Section \ref{sec5}.

\section{Preliminaries}\label{sec1}

Let $F$ be a field and $F[x]$ denote the set of all polynomials with coefficients in $F$ and indeterminate $x$. The highest power of $x$ in a polynomial $f(x)$ is termed the degree of the polynomial, and denoted by $\mbox{deg}\left(f(x)\right)$. All the elements of $F$ can be expressed as polynomials of degree $0$ and are termed scalars. A polynomial of degree $n$ is called monic if the coefficient of $x^n$ is $1$. Denote by $\mbox{gcd}\left(f_1(x),f_2(x),\cdots,f_L(x)\right)$ the greatest common divisor (gcd) of a set of polynomials $f_i(x)$, i.e., the polynomial with the largest degree that divides all of the polynomials $f_i(x)$ for $1\leq i\leq L$. The least common multiple (lcm) of a set of polynomials $f_i(x)$, denoted by $\mbox{lcm}\left(f_1(x),f_2(x),\cdots,f_L(x)\right)$, is the polynomial with the smallest degree that is divisible by every polynomial $f_i(x)$ for $1\leq i\leq L$. For the uniqueness, $\mbox{gcd}(\cdot)$ and $\mbox{lcm}(\cdot)$ are both taken to be monic polynomials. Two polynomials are said to be coprime if their $\mbox{gcd}$ is $1$ or any nonzero scalar in $F$. A polynomial is said to be irreducible if it has only a scalar and itself as its factors. The residue of $f(x)$ modulo $g(x)$ is denoted as $\left[ f(x)\right]_{g(x)}$. Throughout the paper, all polynomials are defined over a field $F$, and $\lfloor\star\rfloor$ and $\lceil\star\rceil$ are well known as the floor and ceiling functions.

Let $m_1(x),m_2(x),\cdots,m_L(x)$ be $L$ non-pairwise coprime polynomial moduli, and $M(x)$ be the lcm of all the moduli, i.e., $M(x)=\mbox{lcm}\left(m_1(x),m_2(x),\cdots,m_L(x)\right)$.
For any polynomial $a(x)$ with $\mbox{deg}\left(a(x)\right)<\mbox{deg}\left(M(x)\right)$, it can be represented by its residue vector $\left(a_1(x),a_2(x),\cdots,a_L(x)\right)$, where $a_i(x)=\left[a(x)\right]_{m_i(x)}$, i.e.,
\begin{equation}\label{folding}
a(x)=k_i(x)m_i(x)+a_i(x)
\end{equation}
with $\mbox{deg}\left(a_i(x)\right)<\mbox{deg}\left(m_i(x)\right)$ and $k_i(x)\in F[x]$ for $1\leq i\leq L$. Here, we call such $k_i(x)$ in (\ref{folding}) the folding polynomials.
Equivalently, $a(x)$ can be computed from its residue vector via the CRT for polynomials \cite{CRT1,ss},
\begin{equation}\label{crtpoly}
a(x)=\left[\sum_{i=1}^{L}a_i(x)D_i(x)M_i(x)\right]_{M(x)},
\end{equation}
where $M_i(x)=\frac{M(x)}{\mu_i(x)}$, $D_i(x)$ is the multiplicative inverse of $M_i(x)$ modulo $\mu_i(x)$, if $\mu_i(x)\neq 1$, else $D_i(x)=0$, and $\{\mu_i(x)\}_{i=1}^{L}$ is a set of $L$ pairwise coprime monic polynomials such that $\prod_{i=1}^{L}\mu_i(x)=M(x)$ and $\mu_i(x)$ divides $m_i(x)$ for each $1\leq i\leq L$. Note that if $m_i(x)$ are pairwise coprime, we have $\mu_i(x)=m_i(x)$ for $1\leq i\leq L$, and then the above reconstruction reduces to the traditional CRT for polynomials.

As seen in the above, polynomials $a(x)$ with $\mbox{deg}\left(a(x)\right)<\mbox{deg}\left(M(x)\right)$ and their residue vectors are isomorphic. Furthermore, the isomorphism holds for the addition, substraction, and multiplication between two polynomials $a(x)$ and $b(x)$, both with degrees less than $\mbox{deg}\left(M(x)\right)$. First convert each polynomial to a residue vector as
\begin{equation}
a(x)\leftrightarrow\left(a_1(x),\cdots,a_L(x)\right)\mbox{ and }b(x)\leftrightarrow\left(b_1(x),\cdots,b_L(x)\right).
\end{equation}
Then, the residue representation of $c(x)=a(x)\pm b(x)$ or $d(x)=[a(x)b(x)]_{M(x)}$ is given by, respectively,
\begin{equation}\label{lin}
c(x)\leftrightarrow\left(a_1(x)\pm b_1(x),\cdots,a_L(x)\pm b_L(x)\right),
\end{equation}
\begin{equation}
d(x)\leftrightarrow\left([a_1(x)b_1(x)]_{m_1(x)},\cdots,[a_L(x)b_L(x)]_{m_L(x)}\right).
\end{equation}
Moreover, an important property in a polynomial remainder code with non-pairwise coprime moduli is that if $a(x)\equiv a_i(x)\mbox{ mod }m_i(x)$ and $a(x)\equiv a_j(x)\mbox{ mod }m_j(x)$, the following congruence holds \cite{CRT1,ss}:
\begin{equation}\label{consi}
a_i(x)\equiv a_j(x)\mbox{ mod }d_{ij}(x),
\end{equation}
where $d_{ij}(x)=\mbox{gcd}\left(m_i(x),m_j(x)\right)$. We call equation (\ref{consi}) a consistency check between residues $a_i(x)$ and $a_j(x)$. If (\ref{consi}) holds, $a_i(x)$ is said to be consistent with $a_j(x)$; otherwise, $a_i(x)$ and $a_j(x)$ appear in a failed consistency check. A residue vector $\left(a_1(x),a_2(x),\cdots,a_L(x)\right)$ is said to be a polynomial remainder codeword if it satisfies the consistency checks given by (\ref{consi}) for all pairs of residues in the vector. So, any polynomial $a(x)$ with $\mbox{deg}\left(a(x)\right)<\mbox{deg}\left(M(x)\right)$ is represented by a unique polynomial remainder codeword, i.e., its residue vector. Conversely, every polynomial remainder codeword is the representation of a unique polynomial with degree less than $\mbox{deg}\left(M(x)\right)$. We call the set of such codewords a polynomial remainder code with moduli $m_1(x),\cdots,m_L(x)$, which is linear according to (\ref{lin}).

If $t$ errors, $e_{i_1}(x),\cdots,e_{i_t}(x)$, in the residues have occurred in the transmission, then the received residue vector, denoted by $\left(\tilde{a}_1(x),\cdots,\tilde{a}_L(x)\right)$, is determined by
\begin{equation}\label{rr}
 \left(\tilde{a}_1(x),\cdots,\tilde{a}_L(x)\right)=\left(a_1(x),\cdots,a_L(x)\right)
 \end{equation}
 \begin{equation*}
 \hspace{2cm}+\left(0,\cdots,e_{i_1}(x),\cdots,e_{i_2}(x),\cdots,e_{i_t}(x),\cdots\right),
\end{equation*}
where $\mbox{deg}\left(e_{i_j}(x)\right)<\mbox{deg}\left(m_{i_j}(x)\right)$ for $1\leq j\leq t$, and the subscripts $i_1,\cdots,i_t$ are the corresponding positions of the residue errors $e_{i_1}(x),\cdots,e_{i_t}(x)$. In \cite{ss}, the capability of residue error correction in a polynomial remainder code with non-pairwise coprime moduli has been investigated, and a simple method for residue error correction has been proposed as well. Before briefly reviewing them, let us present some notations and terminologies in polynomial remainder codes with non-pairwise coprime moduli used in \cite{ss}. Hamming weight of a codeword is the number of nonzero residues in the codeword, Hamming distance between two codewords is defined as the Hamming weight of the difference of the two codewords, and code distance of a polynomial remainder code is the minimum of the Hamming distances between all pairs of different codewords. Due to its linearity (\ref{lin}), the code distance is actually equal to the smallest Hamming weight over all nonzero codewords. Similar to a conventional binary linear code, a polynomial remainder code with code distance $d$ can detect up to $d-1$ errors in the residues, and correct up to $\lfloor(d-1)/2\rfloor$ errors of arbitrary values in the residues. A test for the code distance of a polynomial remainder code with non-pairwise coprime moduli is presented in the following.

\begin{proposition}\cite{ss}\label{p1}
Let $m_i(x)$, $1\leq i\leq L$, be $L$ non-pairwise coprime polynomial moduli, and denote $M(x)=\mbox{lcm}\left(m_1(x),m_2(x),\cdots,m_L(x)\right)$. Write $M(x)$ in the form
\begin{equation}\label{fenie}
M(x)=p_1(x)^{t_1}p_2(x)^{t_2}\cdots p_K(x)^{t_K},
\end{equation}
where the polynomials $p_i(x)$ are pairwise coprime, monic and irreducible, and $t_i$ is a positive integer for all $1\leq i\leq K$. For each $1\leq i\leq K$, let $d_i$ represent the number of moduli that contain the factor $p_i(x)^{t_i}$. Then, the code distance of the polynomial remainder code with the set of moduli $\{m_1(x),m_2(x),\cdots,m_L(x)\}$ is $d=\min\{d_1,d_2,\cdots,d_K\}$.
\end{proposition}

Based on Proposition \ref{p1}, an explicit method of constructing a polynomial remainder code with code distance $d$ is also proposed in \cite{ss}. Let $M(x)$ be decomposed into the product of several smaller, pairwise coprime, and monic polynomials $p_i(x)^{t_i}$ as in the form (\ref{fenie}), $L$ represent the number of moduli in the code, and $d$ be a positive integer such that $1\leq d\leq L$. For each $p_i(x)^{t_i}$, assign $p_i(x)^{t_i}$ to $d_i$ different moduli, such that $d_i\geq d$, with the equality for at least one $i$. Set each modulus to be the product of all polynomials assigned to it. Then, the resulting polynomial remainder code will have the code distance $d$. In particular, repetition codes can be obtained in the above construction by setting $d_i=d=L$, i.e., all moduli are identical. Next, the polynomial remainder code defined in Proposition \ref{p1} enables fast error correction, as described in the following propositions.

\begin{proposition}\cite{ss}\label{p2}
In a polynomial remainder code with code distance $d$ defined in Proposition \ref{p1}, if only $t\leq\lfloor(d-1)/2\rfloor$ errors in the residues have occurred in the transmission, each erroneous residue will appear in at least $\lceil(d-1)/2\rceil+1$ failed consistency checks. In addition, each correct residue will appear in at most $\lfloor(d-1)/2\rfloor$ failed consistency checks.
\end{proposition}

\begin{proposition}\cite{ss}\label{p3}
Let moduli $m_i(x)$ for $1\leq i\leq L$, $M(x)$ and $d$ be defined in Proposition \ref{p1}. Then, the least common multiple of any $L-(d-1)$ moduli is equal to $M(x)$.
\end{proposition}

Based on Propositions \ref{p2}, \ref{p3}, a polynomial remainder code with code distance $d$ can correct up to $\lfloor(d-1)/2\rfloor$ residues errors, i.e., $a(x)$ can be accurately reconstructed from all the error-free residues that can be fast located through consistency checks for all pairs of residues $\tilde{a}_i(x)$ for $1\leq i\leq L$. With the above result, it is not hard to see the following decoding algorithm for polynomial remainder codes with non-pairwise coprime moduli and code distance $d$.
 \begin{enumerate}
   \item[$1)$] Perform the consistency checks by (\ref{consi}) for all pairs of residues $\tilde{a}_i(x)$, $1\leq i\leq L$, in the received residue vector.
   \item[$2)$] Take all of those residues each of which appears in at most $\lfloor(d-1)/2\rfloor$ failed consistency checks. If the number of such residues is zero, i.e., for every $i$ with $1\leq i\leq L$, $\tilde{a}_i(x)$ appears in at least $\lfloor(d-1)/2\rfloor+1$ failed consistency checks, the decoding algorithm fails. Otherwise, go to $3)$.
   \item[$3)$] If all the residues found in $2)$ are consistent with each other, use them to reconstruct $a(x)$ as $\hat{a}(x)$ via the CRT for polynomials in (\ref{crtpoly}). Otherwise, $\hat{a}(x)$ cannot be reconstructed and the decoding algorithm fails.
 \end{enumerate}

According to Propositions \ref{p2}, \ref{p3}, if there are $\lfloor(d-1)/2\rfloor$ or fewer errors in the residues, $a(x)$ can be accurately reconstructed with the above decoding algorithm, i.e., $\hat{a}(x)=a(x)$. However, if more than $\lfloor(d-1)/2\rfloor$ errors have occurred in the residues, the decoding algorithm may fail, i.e., $\hat{a}(x)$ may not be reconstructed, or even though $a(x)$ can be reconstructed as $\hat{a}(x)$, $\hat{a}(x)=a(x)$ may not hold. In the rest of the paper, we assume without loss of generality that the non-pairwise coprime moduli $m_1(x),m_2(x),\cdots,m_L(x)$ are $L$ arbitrarily monic and distinct polynomials with degrees greater than $0$, and the following notations are introduced for simplicity:
\begin{itemize}
\item[$1)$]$d_{ij}(x)=\mbox{gcd}\left(m_i(x),m_j(x)\right)$ for $1\leq i,j\leq L,i\neq j$;
\item[$2)$]$\tau_{ij}=\mbox{deg}\left(d_{ij}(x)\right)$ for $1\leq i,j\leq L,i\neq j$;
\item[$3)$]$\tau_j=\min\limits_{i}\{\tau_{ij},\mbox{ for }1\leq i\leq L, i\neq j\}$;
\item[$4)$]$\Gamma_{ij}(x)=\frac{m_i(x)}{d_{ij}(x)}$ for $1\leq i,j\leq L,i\neq j$;
\item[$5)$]$w^{(i)}$ denotes the code distance of the polynomial remainder code with moduli $\Gamma_{ji}(x)$ for $1\leq j\leq L, j\neq i$, which can be calculated according to Proposition \ref{p1};
\item[$6)$]$n_{(i)}$ denotes the $i$-th smallest element in an array of positive integers $\mathcal{S}=\left\{n_1,n_2,\cdots,n_K\right\}$. It is obvious that $n_{(1)}=\min \mathcal{S}$ and $n_{(K)}=\max \mathcal{S}$. See $\mathcal{S}=\{3,1,2,3,8\}$ for example, and we have $n_{(1)}=1, n_{(2)}=2, n_{(3)}=n_{(4)}=3, n_{(5)}=8$.
\end{itemize}

\section{Robust CRT for polynomials}\label{sec2}

Let $m_i(x)$, $1\leq i\leq L$, be $L$ non-pairwise coprime polynomial moduli, $M(x)$ be the lcm of the moduli, and $d$ be the code distance of the polynomial remainder code with the moduli. As stated in the previous section, a polynomial $a(x)$ with $\mbox{deg}\left(a(x)\right)<\mbox{deg}\left(M(x)\right)$ can be accurately reconstructed from its erroneous residue vector $\left(\tilde{a}_1(x),\cdots,\tilde{a}_L(x)\right)$, if there are $\lfloor(d-1)/2\rfloor$ or fewer errors affecting the residue vector $\left(a_1(x),\cdots,a_L(x)\right)$. Note that the reconstruction of $a(x)$ is accurate but only a few of the residues are allowed to have errors, and most of the residues have to be error-free. In this section, we consider a robust reconstruction problem on which all residues $a_i(x)$ for $1\leq i\leq L$ are allowed to have errors $e_i(x)$ with small degrees.

\begin{definition}[Robust CRT for Polynomials]\footnote{The general robustness is that the reconstruction error is linearly bounded by the error bound $\tau$ of the observation. It is well known that the traditional CRT (with pairwise coprime moduli) is not robust in the sense that a small error in a remainder may cause a large reconstruction error \cite{CRT1,CRT2}.}\label{def1}
A CRT for polynomials is said to be robust with the robustness bound $\tau$ if a reconstruction $\hat{a}(x)$ can be calculated from the erroneous residues $\tilde{a}_i(x)$ for $1\leq i\leq L$ such that $\mbox{deg}(\hat{a}(x)-a(x))\leq\tau$ whenever the residues are affected by errors with degrees upper bounded by $\tau$, i.e., $\mbox{deg}\left(e_i(x)\right)\leq \tau<\mbox{deg}\left(m_i(x)\right)$ for $1\leq i\leq L$.
\end{definition}

This robust reconstruction problem we are interested in is two-fold: one is how we can robustly reconstruct $a(x)$; the other is how large the robustness bound $\tau$ can be for the robustness to hold. The basic idea for the robust CRT for polynomials is to accurately determine one of the folding polynomials. Consider an arbitrary index $j$ with $1\leq j\leq L$. If the folding polynomial $k_{j}(x)$ is accurately determined, a robust estimate of $a(x)$ can then be given by
\begin{align}\label{es}
\begin{split}
\hat{a}(x)&=k_{j}(x)m_{j}(x)+\tilde{a}_{j}(x)\\
&=k_{j}(x)m_{j}(x)+a_{j}(x)+e_{j}(x),
\end{split}					
\end{align}
i.e., $\mbox{deg}\left(\hat{a}(x)-a(x)\right)=\mbox{deg}\left(e_{j}(x)\right)\leq\tau$. Therefore, the problem is to derive conditions under which $k_j(x)$ can be accurately determined from the erroneous residues $\tilde{a}_{i}(x)$ for $1\leq i\leq L$. To do so, we follow the algorithm in \cite{xiao} for integers.

Without loss of generality, we arbitrarily select the first equation or remainder for $i=1$ in (\ref{folding}) as a reference to be subtracted from the other equations for $2\leq i\leq L$, respectively, and we have
\begin{equation}\label{fangchengzu1}
\left\{\begin{array}{ll}
k_1(x)m_1(x)-k_2(x)m_2(x)=a_2(x)-a_1(x)\\
k_1(x)m_1(x)-k_3(x)m_3(x)=a_3(x)-a_1(x)\\
\:\:\:\:\vdots\\
k_1(x)m_1(x)-k_L(x)m_L(x)=a_L(x)-a_1(x).
\end{array}\right.
\end{equation}
Denote $q_{i1}(x)=\frac{a_i(x)-a_1(x)}{d_{1i}(x)}$ for $2\leq i\leq L$. Then, dividing $d_{1i}(x)$ from both sides of the $(i-1)$-th equation in (\ref{fangchengzu1}) for $2\leq i\leq L$, we can equivalently write (\ref{fangchengzu1}) as
\begin{equation}\label{fangchengzu2}
\left\{\begin{array}{ll}
k_1(x)\Gamma_{12}(x)-k_2(x)\Gamma_{21}(x)=q_{21}(x)\\
k_1(x)\Gamma_{13}(x)-k_3(x)\Gamma_{31}(x)=q_{31}(x)\\
\:\:\:\:\vdots\\
k_1(x)\Gamma_{1L}(x)-k_L(x)\Gamma_{L1}(x)=q_{L1}(x).
\end{array}\right.
\end{equation}
Since $\Gamma_{1i}(x)$ and $\Gamma_{i1}(x)$ are coprime, by B\'{e}zout's lemma for polynomials we have
\begin{equation}\label{kone}
k_1(x)=q_{i1}(x)\bar{\Gamma}_{1i}(x)+k(x)q_{i1}(x)\Gamma_{i1}(x), \mbox{ for } 2\leq i\leq L,
\end{equation}
where $k(x)$ is some polynomial in $F[x]$, and $\bar{\Gamma}_{1i}(x)$ is the multiplicative inverse of $\Gamma_{1i}(x)$ modulo $\Gamma_{i1}(x)$, i.e., $\bar{\Gamma}_{1i}(x)\Gamma_{1i}(x)\equiv1\mbox{ mod }\Gamma_{i1}(x)$.

Next, we can use
\begin{equation}\label{qest}
\begin{array}{lll}
\hat{q}_{i1}(x)&=&\displaystyle\frac{\tilde{a}_{i}(x)-\tilde{a}_1(x)-\left[\tilde{a}_{i}(x)-\tilde{a}_1(x)\right]_{d_{1i}(x)}}{d_{1i}(x)}\\
&=&\displaystyle q_{i1}(x)+\frac{e_{i}(x)-e_1(x)-\left[ e_{i}(x)-e_1(x)\right]_{d_{1i}(x)}}{d_{1i}(x)}
\end{array}
\end{equation}
as an estimate of $q_{i1}(x)$ for $2\leq i\leq L$ in (\ref{kone}), and we have the following algorithm.

\hrulefill

\textit{Algorithm \uppercase\expandafter{\romannumeral1}:}

\hrulefill
\begin{itemize}
  \item \textbf{Step 1:} Calculate $d_{1i}(x)=\mbox{gcd}\left(m_1(x), m_i(x)\right)$, $\Gamma_{1i}(x)=\frac{m_1(x)}{d_{1i}(x)}$, and $\Gamma_{i1}(x)=\frac{m_i(x)}{d_{1i}(x)}$ for $2 \leq i\leq L$ from the given moduli $m_j(x)$ for $1 \leq j\leq L$, which
can be done in advance.
  \item \textbf{Step 2:} Calculate $\hat{q}_{i1}(x)$ for $2 \leq i\leq L$ in (\ref{qest}) from $m_j(x)$ and the erroneous residues $\tilde{a}_j(x)$ for $1 \leq j\leq L$.
  \item \textbf{Step 3:} Calculate the remainders of $\hat{q}_{i1}(x)\bar{\Gamma}_{1i}(x)$ modulo $\Gamma_{i1}(x)$, i.e.,
  \begin{equation}\label{xixi}
   \hat{\xi}_{i1}(x) \equiv \hat{q}_{i1}(x)\bar{\Gamma}_{1i}(x) \mbox{ mod } \Gamma_{i1}(x)
  \end{equation}
  for $2 \leq i\leq L$, where $\bar{\Gamma}_{1i}(x)$ is the multiplicative inverse of $\Gamma_{1i}(x)$ modulo $\Gamma_{i1}(x)$ and can be calculated in advance.
  \item \textbf{Step 4:} Calculate $\hat{k}_{1}(x)$ from the following system of congruences:
  \begin{equation}\label{gcrt}
  \hat{k}_1(x)\equiv \hat{\xi}_{i1}(x) \mbox{ mod } \Gamma_{i1}(x), \mbox{ for } 2 \leq i\leq L,
  \end{equation}
  where moduli $\Gamma_{i1}(x)$ may not be pairwise coprime. Note that $\hat{k}_1(x)$ is calculated by using the decoding algorithm for the polynomial remainder code with moduli $\Gamma_{i1}(x)$ for $2\leq i\leq L$, based on Propositions \ref{p2}, \ref{p3} in Section \ref{sec1}.
\end{itemize}

\hrulefill

\begin{remark}
As we mentioned before, the basic idea in the robust CRT for polynomials is to accurately determine one of folding polynomials, which is different from the robust CRT for integers \cite{wenjie,xiao} where all folding integers are accurately determined and each determined folding integer provides a reconstruction, and all the reconstructions from all the determined folding integers can then be averaged to provide a better estimate. Accordingly, since we do not need to calculate other folding polynomials $k_i(x)$ for $2\leq i\leq L$ in the above \textit{Algorithm \uppercase\expandafter{\romannumeral1}}, $\hat{q}_{i1}(x)=q_{i1}(x)$ in (\ref{qest}) does not have to hold for all $2\leq i\leq L$, that is, residues $\hat{\xi}_{i1}(x)$ in (\ref{gcrt}), $2\leq i\leq L$, are allowed to have a few errors. This is why we use the decoding algorithm in Section \ref{sec1} to reconstruct $k_1(x)$ in Step $4$ in terms of the polynomial remainder code with moduli $\Gamma_{i1}(x)$ for $2\leq i\leq L$.
\end{remark}

Let $w^{(1)}$ denote the code distance of the polynomial remainder code with moduli $\Gamma_{i1}(x)$ for $2\leq i\leq L$ and $\tau$ be the robustness bound, i.e., $\mbox{deg}\left(e_i(x)\right)\leq\tau$ for $1\leq i\leq L$. With the above algorithm, we have the following lemma.
\begin{lemma}\label{lemma2}
$k_1(x)$ can be accurately determined in \textit{Algorithm \uppercase\expandafter{\romannumeral1}}, i.e., $k_1(x)=\hat{k}_1(x)$,
if the robustness bound $\tau$ satisfies
\begin{equation}\label{con2}
\tau<\tau_{1(\lfloor(w^{(1)}-1)/2\rfloor+1)},
\end{equation}
where $\tau_{ij}=\mbox{deg}\left(d_{ij}(x)\right)$ for $1\leq i,j\leq L,i\neq j$, and $\tau_{1(j)}$ denotes the $j$-th smallest element in $\{\tau_{12},\tau_{13},\cdots,\tau_{1L}\}$.
\end{lemma}
\begin{proof}
Let $\Gamma(x)=\mbox{lcm}\left(\Gamma_{21}(x),\Gamma_{31}(x),\cdots,\Gamma_{L1}(x)\right)$. We first prove $\mbox{deg}\left(k_1(x)\right)<\mbox{deg}\left(\Gamma(x)\right)$. If $\mbox{deg}\left(k_1(x)\right)=0$, it is obvious for $\mbox{deg}\left(k_1(x)\right)<\mbox{deg}\left(\Gamma(x)\right)$. If $\mbox{deg}\left(k_1(x)\right)\neq0$, we have $\mbox{deg}\left(k_{1}(x)m_{1}(x)\right)=\mbox{deg}\left(k_{1}(x)\right)+\mbox{deg}\left(m_{1}(x)\right)= \mbox{deg}\left(a(x)\right)<\mbox{deg}\left(M(x)\right)=\mbox{deg}\left(m_1(x)\Gamma(x)\right)=\mbox{deg}\left(m_1(x)\right)+\mbox{deg}\left(
\Gamma(x)\right)$, and thus $\mbox{deg}\left(k_1(x)\right)<\mbox{deg}\left(\Gamma(x)\right)$.

Among the residue errors $e_i(x)$ for $2\leq i\leq L$, there are $v\geq L-1-\lfloor\frac{w^{(1)}-1}{2}\rfloor$ errors $e_{i_j}(x)$ for $1\leq j\leq v$ such that $\mbox{deg}\left(e_{i_j}(x)-e_{1}(x)\right)\leq\tau<\mbox{deg}\left(d_{1i_j}(x)\right)$ according to (\ref{con2}). So, we have $[e_{i_j}(x)-e_{1}(x)]_{d_{1i_j}(x)}=e_{i_j}(x)-e_{1}(x)$. From (\ref{qest}), we have $\hat{q}_{i_j1}(x)=q_{i_j1}(x)$, and it is not hard to see from (\ref{kone}) that $k_1(x)$ and $\hat{q}_{i_j1}(x)\bar{\Gamma}_{1i_j}(x)$ have the same remainder modulo $\Gamma_{i_j1}(x)$, i.e., $k_1(x)\equiv \hat{\xi}_{i_j1}(x) \mbox{ mod } \Gamma_{i_j1}(x)$. Then, regard $(\hat{\xi}_{21}(x),\hat{\xi}_{31}(x),\cdots,\hat{\xi}_{L1}(x))$ in (\ref{gcrt}) as an erroneous residue vector of $k_1(x)$ modulo $\Gamma_{i1}(x)$ for $2\leq i\leq L$. Since there are at most $\lfloor(w^{(1)}-1)/2\rfloor$ errors in $(\hat{\xi}_{21}(x),\hat{\xi}_{31}(x),\cdots,\hat{\xi}_{L1}(x))$, we can accurately determine the folding polynomial $k_1(x)$ in Step $4$ of \textit{Algorithm \uppercase\expandafter{\romannumeral1}}, i.e., $\hat{k}_1(x)=k_1(x)$, by applying the residue error correction algorithm based on Propositions \ref{p2}, \ref{p3} for the polynomial remainder code with moduli $\Gamma_{i1}(x)$ for $2\leq i\leq L$.
\end{proof}

Recall that $m_1(x)$ or $a_1(x)$ in the above \textit{Algorithm \uppercase\expandafter{\romannumeral1}} is arbitrarily selected to be a reference, which is not necessary. In fact, any remainder can be taken as the reference. In order to improve the maximal possible robustness bound, we next present the following theorem through selecting a proper reference folding polynomial.

\begin{theorem}\label{cor2}
If the robustness bound $\tau$ satisfies
\begin{equation}\label{ccc}
\tau< \max_{1\leq i\leq L}\left\{\tau_{i(\lfloor(w^{(i)}-1)/2\rfloor+1)}\right\},
\end{equation}
where $w^{(i)}$ is the code distance of the polynomial remainder code with moduli $\Gamma_{ji}(x)$ for $1\leq j\leq L, j\neq i$, and $\tau_{i(j)}$ denotes the $j$-th smallest element in $\{\tau_{ik},\mbox{ for }1\leq k\leq L, k\neq i\}$, then $a(x)$ can be robustly reconstructed through \textit{Algorithm \uppercase\expandafter{\romannumeral1}}, that is, the robust CRT for polynomials in Definition \ref{def1} holds.
\end{theorem}
\begin{proof}
Let us choose such an index $i_0$ that
\begin{equation}\label{refer}
\tau_{i_0(\lfloor(w^{(i_0)}-1)/2\rfloor+1)}=\max_{1\leq i\leq L}\left\{\tau_{i(\lfloor(w^{(i)}-1)/2\rfloor+1)}\right\}.
\end{equation}
Then, replacing the index $1$ with $i_0$ and taking $\tilde{a}_{i_0}(x)$ as the reference in \textit{Algorithm \uppercase\expandafter{\romannumeral1}}, we can accurately determine $k_{i_0}(x)$ under the condition (\ref{ccc}), thereby robustly reconstructing $a(x)$ as $\hat{a}(x)$ in (\ref{es}), i.e., $\mbox{deg}\left(\hat{a}(x)-a(x)\right)\leq\tau$.
\end{proof}

\begin{remark}
From (\ref{refer}), it guarantees that there are at most $\lfloor(w^{(i_0)}-1)/2\rfloor$ errors in the residues $\hat{\xi}_{ii_0}(x)$ for $1\leq i\leq L,i\neq i_0$ in Step $4$ of \textit{Algorithm \uppercase\expandafter{\romannumeral1}} to determine $k_{i_0}(x)$ with respect to moduli $\Gamma_{ii_0}(x)$. If $w^{(i_0)}\geq3$, $k_{i_0}(x)$ is accurately determined based on the residue error correction algorithm in Section \ref{sec1} for the polynomial remainder code with moduli $\Gamma_{ii_0}(x)$ for $1\leq i\leq L,i\neq i_0$. If $w^{(i_0)}<3$, all the residues $\hat{\xi}_{ii_0}(x)$ for $1\leq i\leq L,i\neq i_0$, are accurate and consistent, and $k_{i_0}(x)$ is accurately determined via the CRT for polynomials from all the $L-1$ residues $\hat{\xi}_{ii_0}(x)$.
\end{remark}

\begin{example}
Let us consider a notable class of polynomial remainder codes with special moduli (the corresponding integer residue codes were introduced in \cite{bm,chessa}), i.e., $m_i(x)=\prod\limits_{j\in[1,L];j\neq i}d_{ij}(x)$ for $1\leq i\leq L$ and $\{d_{ij}(x), \mbox{for }1\leq i\leq L; i<j\leq L\}$ are pairwise coprime. Let $d_{12}(x)=(x+1)^4, d_{13}(x)=(x-1)^4, d_{14}(x)=(x+2)^4, d_{23}(x)=(x-2)^4, d_{24}(x)=(x+3)^4, d_{34}(x)=(x-3)^4$. Since $\mbox{deg}\left(d_{ij}(x)\right)=4$ holds for every $1\leq i,j\leq 4, i\neq j$, it is easy to see from Theorem \ref{cor2} that the robustness bound is $\tau<4$, i.e., any $a(x)$ with $\mbox{deg}\left(a(x)\right)<\mbox{deg}\left(\mbox{lcm}(m_1(x),\cdots,m_4(x))\right)=24$ can be robustly reconstructed from its erroneous residues when the degrees of all residue errors are less than $4$.
\end{example}

If the above result is referred to as the single stage robust CRT for polynomials, multi-stage robust CRT for polynomials can be easily derived by following the method used for integers in \cite{xiao}. Similarly, multi-stage robust CRT for polynomials may improve the bound for $\tau$ obtained in Theorem \ref{cor2} for a given set of polynomial moduli. Another remark we make here is that a residue error $e_i(x)$ is said to be a bounded error with an error bound $l$ if its degree is less than or equal to $l$, where $l$ is a small positive integer. What Theorem \ref{cor2} tells us is that for the set of moduli $m_i(x)$ in the above, a polynomial $a(x)$ with $\mbox{deg}\left(a(x)\right)<\mbox{deg}\left(M(x)\right)$ can be robustly reconstructed from its erroneous residues if all residue errors are bounded, and the error bound $\tau$ is given by (\ref{ccc}). Later, the constraint that all residue errors are bounded will be released, and the combined occurrence of multiple unrestricted errors and an arbitrary number of bounded errors in the residues will be considered in the next section.

\section{Robust reconstruction under multiple unrestricted errors and an arbitrary number of bounded errors in the residues}\label{sec4}

Consider again the $L$ non-pairwise coprime moduli $m_i(x)$ for $1\leq i\leq L$. In this section, we assume that there are $t\leq\lfloor(d-1)/2\rfloor$ unrestricted errors and an arbitrary number of bounded errors with the error bound $\lambda$ in the received residue vector $\left(\tilde{a}_1(x),\cdots,\tilde{a}_L(x)\right)$. Similarly in this case, the robust reconstruction problems for us are: 1) how can we robustly reconstruct $a(x)$? 2) how large can the error bound $\lambda$ be for the robustness to hold? Note that $d\geq3$ is necessarily assumed in this section, otherwise it is degenerated to the case of robust CRT for polynomials in Section \ref{sec2}. Therefore, due to the existence of unrestricted residue errors, the bound for $\lambda$ is expected to be smaller than or equal to the bound for $\tau$ as in (\ref{ccc}). In order to answer the above questions, we first give the following lemmas.

\begin{lemma}\label{lemma1}
Let $w^{(i)}$ denote the code distance of the polynomial remainder code with moduli $\Gamma_{ji}(x)$ for $1\leq j\leq L, j\neq i$. Then, we have \begin{equation}
\min\{w^{(1)},w^{(2)},\cdots,w^{(L)}\}=d,
\end{equation}
where $d$ is the code distance of the polynomial remainder code with moduli $m_i(x)$ for $1\leq i\leq L$.
\end{lemma}
\begin{proof}
First, let us prove $w^{(i)}\geq d$ for each $1\leq i\leq L$. Without loss of generality, we only need to prove $w^{(1)}\geq d$.
Let $M(x)$ be written as in (\ref{fenie}), i.e.,
\begin{equation}
M(x)=p_1(x)^{t_1}p_2(x)^{t_2}\cdots p_K(x)^{t_K},
\end{equation}
where the polynomials $p_i(x)$ are pairwise coprime, monic and irreducible, and $t_i$ is a positive integer for all $1\leq i\leq K$. Define $\Gamma(x)=\mbox{lcm}\left(\Gamma_{21}(x),\Gamma_{31}(x),\cdots,\Gamma_{L1}(x)\right)$. Since $M(x)=m_1(x)\Gamma(x)$, we can write $m_1(x)$ and $\Gamma(x)$ as
\begin{multline}
\;\;\;\;\;\;\;\;\;\,m_1(x)=p_1(x)^{l_1}p_2(x)^{l_2}\cdots p_K(x)^{l_K}\mbox{ and } \\
\Gamma(x)=p_1(x)^{t_1-l_1}p_2(x)^{t_2-l_2}\cdots p_K(x)^{t_K-l_K},
\end{multline}
where $t_i\geq l_i\geq 0$ for $1\leq i\leq K$. First, consider $0\leq l_i<t_i$, and let $w_i^{(1)}$ represent the number of moduli $\Gamma_{21}(x),\cdots,\Gamma_{L1}(x)$ that contain the factor $p_i(x)^{t_i-l_i}$. In this case, we have $w_i^{(1)}=d_i$, where $d_i$ is defined in Proposition \ref{p1}, because for every $j$ with $2\leq j\leq L$, $m_j(x)$ contains $p_i(x)^{t_i}$ if and only if its corresponding $\Gamma_{j1}(x)=\frac{m_j(x)}{d_{1j}(x)}$ contains the factor $p_i(x)^{t_i-l_i}$. Next, consider $l_i=t_i$, and $\Gamma(x)$ does not contain the item of $p_i(x)$. Hence, according to Proposition \ref{p1}, $w^{(1)}$ is the minimum of $\{d_1,d_2,\cdots,d_K\}$, i.e., $w^{(1)}=d$, if $0\leq l_i<t_i$ holds for all $1\leq i\leq K$, else $w^{(1)}$ is the minimum of a subset of $\{d_1,d_2,\cdots,d_K\}$, i.e., $w^{(1)}\geq d$. So, we have $w^{(1)}\geq d$. Note that the above proof is independent of an arbitrary choice $i=1$ for $w^{(i)}$. Therefore, we have $w^{(i)}\geq d$ for each $1\leq i\leq L$.

Next, we prove that there is at least one $i$ such that $w^{(i)}=d$. Without loss of generality, we assume that $d_1=d$. From the above analysis, if $w^{(1)}>d$, we must have $l_1=t_1$, i.e., $m_1(x)$ contains the factor $p_1(x)^{t_1}$. Similarly, if all $w^{(i)}$ for $1\leq i\leq L$ are strictly larger than $d$, we know that all $m_i(x)$ for $1\leq i\leq L$ contain the factor $p_1(x)^{t_1}$, i.e., $d_1=d=L$. Thus, $d_i=L$ for all $1\leq i\leq L$. It is in contradiction with the assumption in the end of Section \ref{sec1} that $m_1(x),m_2(x),\cdots,m_L(x)$ are monic and distinct polynomials with degrees greater than $0$. Thus, we have $\min\{w^{(1)},\cdots,w^{(L)}\}=d$.
\end{proof}

\begin{lemma}\label{31}
Let $d$ denote the code distance of the polynomial remainder code with moduli $m_i(x)$ for $1\leq i\leq L$. Assume that there are $t\leq\lfloor(d-1)/2\rfloor$ unrestricted errors, and any other error is bounded in the received residue vector $\left(\tilde{a}_1(x),\cdots,\tilde{a}_L(x)\right)$. The error bound $\lambda$ here is assumed less than $\tau_1$, where $\tau_1=\min\limits_{j}\{\tau_{1j},\mbox{ for }2\leq j\leq L\}$ and $\tau_{1j}=\mbox{deg}\left(d_{1j}(x)\right)$. If $\tilde{a}_1(x)$ is known as an error-free residue or a residue with a bounded error, we can accurately determine $k_1(x)$ using \textit{Algorithm \uppercase\expandafter{\romannumeral1}}, i.e., $\hat{k}_1(x)=k_1(x)$. However, if $\tilde{a}_1(x)$ is known as a residue with an error of degree greater than $\lambda$, $\hat{k}_1(x)$ may not be reconstructed, and even though $\hat{k}_1(x)$ is reconstructed in \textit{Algorithm \uppercase\expandafter{\romannumeral1}}, $\hat{k}_1(x)=k_1(x)$ may not hold.
\end{lemma}
\begin{proof}
If $\tilde{a}_i(x)$ with $i>1$ is an error-free residue or a residue with a bounded error, i.e., $e_i(x)=0$ or $e_i(x)\neq 0\mbox{ with deg}\left(e_i(x)\right)\leq\lambda$, we have $e_i(x)-e_1(x)=[e_i(x)-e_1(x)]_{d_{1i}(x)}$. This is due to the fact that $\mbox{deg}\left(e_i(x)-e_1(x)\right)\leq\lambda<\tau_1\leq\mbox{deg}\left(d_{1i}(x)\right)$. Therefore, we have $k_1(x)\equiv\hat{\xi}_{i1}(x)\mbox{ mod }\Gamma_{i1}(x)$ from (\ref{kone}), (\ref{qest}), and (\ref{xixi}). Since there are only $t$ unrestricted residue errors and any other error is bounded  in the residues, there are at most $t$ residue errors with degrees greater than $\lambda$. In other words, there are at least $L-1-t$ residues $\tilde{a}_i(x)$ with $i\neq1$ that are error-free or with bounded errors. Therefore, there are at most $t$ errors in $(\hat{\xi}_{21}(x),\hat{\xi}_{31}(x),\cdots,\hat{\xi}_{L1}(x))$ to calculate $k_1(x)$ in Step $4$ of \textit{Algorithm \uppercase\expandafter{\romannumeral1}}.
Due to $t\leq\lfloor(d-1)/2\rfloor$ and $d\leq w^{(1)}$, $k_1(x)$ can be accurately determined in \textit{Algorithm \uppercase\expandafter{\romannumeral1}} by applying the residue error correction algorithm based on Propositions \ref{p2}, \ref{p3} for the polynomial remainder code with moduli $\Gamma_{i1}(x)$ for $2\leq i\leq L$, i.e., $\hat{k}_1(x)=k_1(x)$.

However, if $\tilde{a}_1(x)$ is known as a residue with an error of degree greater than $\lambda$, it is not guaranteed that there are at most $\lfloor(w^{(1)}-1)/2\rfloor$ errors in $(\hat{\xi}_{21}(x),\cdots,\hat{\xi}_{L1}(x))$ in Step $4$ of \textit{Algorithm \uppercase\expandafter{\romannumeral1}}. Therefore, following the decoding algorithm in Section \ref{sec1}, $\hat{k}_1(x)$ may not be reconstructed, and even though $\hat{k}_1(x)$ is reconstructed, $\hat{k}_1(x)=k_1(x)$ may not hold.
\end{proof}

From the above results, we have the following theorem.
\begin{theorem}\label{t2}
Let $m_i(x)$, $1\leq i\leq L$, be $L$ non-pairwise coprime moduli, $M(x)$ be the lcm of the moduli, and the erroneous residue vector of $a(x)$ with $\mbox{deg}\left(a(x)\right)<\mbox{deg}\left(M(x)\right)$ be denoted as $\left(\tilde{a}_1(x),\tilde{a}_2(x),\cdots,\tilde{a}_L(x)\right)$. Denote by $d$ the code distance of the polynomial remainder code with moduli $m_i(x)$ for $1\leq i\leq L$. Assume that there are $t\leq\lfloor(d-1)/2\rfloor$ unrestricted errors and an arbitrary number of bounded errors in the residues. Then, if the remainder error bound $\lambda$ satisfies
\begin{equation}\label{gi}
\lambda<\tau_{(L-2\lfloor(d-1)/2\rfloor)},
\end{equation}
where $\tau_{(i)}$ denotes the $i$-th smallest element in $\{\tau_1,\cdots,\tau_L\}$, each $\tau_j$ for $1\leq j\leq L$ is defined as $\tau_j=\min\limits_{i}\{\tau_{ij},\mbox{ for }1\leq i\leq L, i\neq j\}$, and $\tau_{ij}=\mbox{deg}\left(d_{ij}(x)\right)$,
we can robustly reconstruct $a(x)$ as $\hat{a}(x)$, i.e., $\mbox{deg}\left(a(x)-\hat{a}(x)\right)\leq\lambda$, by following \textit{Algorithm \uppercase\expandafter{\romannumeral1}}.
\end{theorem}
\begin{proof}
Without loss of generality, we assume $\tau_1\geq\tau_2\geq\cdots\geq\tau_L$. First, by taking every residue in the first $2\lfloor(d-1)/2\rfloor+1$ residues as a reference and following \textit{Algorithm \uppercase\expandafter{\romannumeral1}}, we want to calculate the corresponding folding polynomial $\hat{k}_i(x)$, respectively. When $\tilde{a}_i(x)$ for $1\leq i\leq2\lfloor(d-1)/2\rfloor+1$ is known as an error-free residue or a residue with a bounded error and the bound is $\lambda$, since $\lambda<\tau_{(L-2\lfloor(d-1)/2\rfloor)}\leq\tau_i$ from (\ref{gi}), it follows from Lemma \ref{31} that $k_{i}(x)$ can be accurately determined by \textit{Algorithm \uppercase\expandafter{\romannumeral1}}, i.e., $\hat{k}_i(x)=k_i(x)$. Since there are at most $\lfloor(d-1)/2\rfloor$ residues with errors of degrees greater than $\lambda$ in the first $2\lfloor(d-1)/2\rfloor+1$ residues, there are at least $\lfloor(d-1)/2\rfloor+1$ error-free residues or residues with bounded errors. Therefore, at least $\lfloor(d-1)/2\rfloor+1$ folding polynomials out of $\hat{k}_i(x)$ for $1\leq i\leq2\lfloor(d-1)/2\rfloor+1$ are accurately determined. However, when $\tilde{a}_i(x)$ is a residue with an error of degree greater than $\lambda$ and taken as a reference, the corresponding folding polynomial $\hat{k}_i(x)$ may not be reconstructed in \textit{Algorithm \uppercase\expandafter{\romannumeral1}}, and even though $\hat{k}_i(x)$ is reconstructed, it may not be equal to $k_i(x)$.

Then, for each obtained $\hat{k}_i(x)$ with $1\leq i\leq 2\lfloor(d-1)/2\rfloor+1$, we reconstruct $a(x)$ as $\hat{a}^{[i]}(x)=\hat{k}_i(x)m_i(x)+\tilde{a}_i(x)$ for $1\leq i\leq 2\lfloor(d-1)/2\rfloor+1$. If $\tilde{a}_i(x)$ and $\tilde{a}_j(x)$ with $i\neq j$ are both error-free residues or residues with bounded errors, we have $\mbox{deg}\left(\hat{a}^{[i]}(x)-a(x)\right)\leq\lambda$ and $\mbox{deg}\left(\hat{a}^{[j]}(x)-a(x)\right)\leq\lambda$, since $\hat{k}_i(x)$ and $\hat{k}_j(x)$ are accurately determined. Thus, we have  $\mbox{deg}\left(\hat{a}^{[i]}(x)-\hat{a}^{[j]}(x)\right)\leq\lambda$. If $\tilde{a}_i(x)$ is a residue with an error of degree greater than $\lambda$ and $\hat{k}_i(x)$ is reconstructed, one can see that no matter whether $\hat{k}_i(x)$ is accurate or not, we will have $\mbox{deg}\left(\hat{a}^{[i]}(x)-a(x)\right)>\lambda$. This is due to the fact that $\hat{a}^{[i]}(x)-a(x)=(\hat{k}_i(x)-k_i(x))m_i(x)+\left(\tilde{a}_i(x)-a_i(x)\right)$ and $\mbox{deg}\left(m_i(x)\right)>\mbox{deg}\left(\tilde{a}_i(x)-a_i(x)\right)>\lambda$. Furthermore, we can easily obtain $\mbox{deg}\left(\hat{a}^{[i]}(x)-\hat{a}^{[j]}(x)\right)>\lambda$ when one of the corresponding references $\tilde{a}_i(x)$ and $\tilde{a}_j(x)$ used for reconstruction in \textit{Algorithm \uppercase\expandafter{\romannumeral1}} is a residue with an error of degree greater than $\lambda$, and the other is an error-free residue or a residue with a bound error.

Therefore, among the above at most $2\lfloor(d-1)/2\rfloor+1$ reconstructions $\hat{a}^{[i]}(x)$, we can find at least $\lfloor(d-1)/2\rfloor+1$ reconstructions such that $\mbox{deg}\left(\hat{a}^{[i]}(x)-\hat{a}^{[j]}(x)\right)\leq\lambda$ among pairs of $i,j$ with $1\leq i\neq j\leq 2\lfloor(d-1)/2\rfloor+1$. One can see that all of such reconstructions are in fact obtained when references $\tilde{a}_i(x)$ are error-free residues or residues with bounded errors, and thus, any one of such reconstructions can be thought of as a robust reconstruction of $a(x)$. At this point, we have completed the proof.
\end{proof}

According to the above proof of Theorem \ref{t2}, let us summarize the robust reconstruction algorithm for a given set of moduli $\{m_i(x)\}_{i=1}^{L}$, with which the polynomial remainder code has the code distance $d$. Assume that $\tau_1\geq\tau_2\geq\cdots\geq\tau_L$ and there are $t\leq\lfloor(d-1)/2\rfloor$ unrestricted errors and an arbitrary number of bounded errors with the error bound $\lambda$ given by (\ref{gi}) in the residues.
\begin{enumerate}
   \item For every $i$ with $1\leq i\leq 2\lfloor(d-1)/2\rfloor+1$, take $\tilde{a}_i(x)$ as a reference and follow \textit{Algorithm \uppercase\expandafter{\romannumeral1}}. We want to calculate the corresponding $\hat{k}_i(x)$ for $1\leq i\leq 2\lfloor(d-1)/2\rfloor+1$, respectively. Note that some $\hat{k}_i(x)$ may not be reconstructed.
   \item Reconstruct $a(x)$ as $\hat{a}^{[i]}(x)=\hat{k}_i(x)m_i(x)+\tilde{a}_i(x)$ for each obtained $\hat{k}_i(x)$.
   \item Among these obtained $\hat{a}^{[i]}(x)$, we can find at least $\lfloor(d-1)/2\rfloor+1$ reconstructions $\hat{a}^{[i_j]}(x)$ for $1\leq j\leq \mu$ with $\mu\geq\lfloor(d-1)/2\rfloor+1$ such that every pair of them satisfy $\mbox{deg}\left(\hat{a}^{[i_{{\varsigma}}]}(x)-\hat{a}^{[i_{{\varrho}}]}(x)\right)\leq\lambda$ for $1\leq \varsigma, \varrho\leq\mu$, $\varsigma\neq\varrho$. Then, any one of such $\hat{a}^{[i_j]}(x)$ for $1\leq j\leq \mu$ can be regarded as a robust reconstruction of $a(x)$, i.e., $\mbox{deg}\left(\hat{a}^{[i_j]}(x)-a(x)\right)\leq\lambda$.
 \end{enumerate}

\begin{remark}\label{houm}
There is a related paper dealing with robustly reconstructing an integer from erroneous remainders \cite{chessa}, but our paper investigating the robust reconstruction problems for polynomials differs from \cite{chessa} in several aspects as follows:
\begin{itemize}
  \item[$a)$] The problem of robust reconstruction for integers from erroneous residues was considered in \cite{chessa} and \cite{wenjie,xiao} with different approaches. In \cite{chessa}, a large integer is robustly reconstructed through constructing a new consistent residue vector from the erroneous residues. In \cite{wenjie,xiao}, however, all folding integers are first accurately determined, and then a robust reconstruction is provided as an average of all the reconstructions from all the determined folding integers. In this paper, an improved reconstruction algorithm for polynomials in \textit{Algorithm \uppercase\expandafter{\romannumeral1}} is proposed by combining the approach in \cite{xiao} with the error correction algorithm for polynomial remainder codes in \cite{ss}. While both of the approaches in \cite{chessa} and \cite{xiao} can be directly extended to robust reconstruction for polynomials in Section \ref{sec2}, the obtained maximal possible robustness bounds would be usually less than (\ref{ccc}) obtained in our proposed algorithm.
  \item[$b)$] In \cite{chessa}, a special class of residue number systems with non-pairwise coprime moduli was only considered, where moduli $m_i=\prod\limits_{j\in[1,L];j\neq i}d_{ij}$ for $1\leq i\leq L$, $d_{ij}=d_{ji}$, and $\{d_{ij}, \mbox{for }1\leq i\leq L; i<j\leq L\}$ are pairwise coprime and greater than $1$. According to Proposition \ref{p1} for integers, one can see that this residue code with these moduli $m_i$ for $1\leq i\leq L$ has code distance $2$, which is unable to correct any residue errors. So, in order to enable single errors to be corrected, the legitimate range of the code must be restricted to a suitable subrange of $[0,\mbox{lcm}(m_1,\cdots,m_L))$ in \cite{bm}, and in \cite{chessa}, robust reconstruction for the case of a combined occurrence of a single unrestricted error and an arbitrary number of small errors in the residues was considered also with the legitimate range being a suitable subrange of $[0,\mbox{lcm}(m_1,\cdots,m_L))$. Its approach is hard to deal with the case of multiple unrestricted errors combined with small errors in the residues due to a considerable decoding complexity. In this paper, however, we consider the robust reconstruction problem for polynomials from the perspective of polynomial remainder codes with non-pairwise coprime moduli. The range of the degree of $a(x)$ is fixed for a general set of moduli $\{m_i(x)\}_{i=1}^{L}$, i.e., $\mbox{deg}\left(a(x)\right)<\mbox{deg}\left(\mbox{lcm}\left(m_1(x),\cdots,m_L(x)\right)\right)$. Under the assumption that the code distance of the polynomial remainder code with moduli $m_i(x)$ for $1\leq i\leq L$ is $d\geq3$, the remainder error bound and/or the maximum possible number of unrestricted residue errors are obtained for the robustness to hold in the paper. Moreover, a well-established algorithm based on Theorem \ref{t2} is proposed to robustly reconstruct a polynomial when there are multiple unrestricted errors and an arbitrary number of bounded errors in the residues, where \textit{Algorithm \uppercase\expandafter{\romannumeral1}} needs to be implemented $2\lfloor(d-1)/2\rfloor+1$ times.
  \item[$c)$] Compared with \cite{chessa}, we omit to consider the case of erasures in the residues in this paper due to its triviality. As we know, when erasures occur in the residues, both the number and positions of erasures are known. Without loss of generality, assume that $\aleph$ erasures occur and the erased residues are $a_{L-\aleph+1}(x),\cdots,a_L(x)$. So, we just calculate the code distance of the polynomial remainder code with moduli $m_i(x)$ for $1\leq i\leq L-\aleph$ according to Proposition \ref{p1} and consider to robustly reconstruct $a(x)$ with $\mbox{deg}\left(a(x)\right)<\mbox{deg}\left(\mbox{lcm}\left(m_1(x),\cdots,m_{L-\aleph}(x)\right)\right)$ from those available $\tilde{a}_i(x)$ for $1\leq i\leq L-\aleph$ in Theorem \ref{cor2} and Theorem \ref{t2}.
\end{itemize}
\end{remark}

\begin{example}\label{exaa}
Let $L=5$ and the moduli be $m_1(x)=(x^3+1)(x^2-2)(x^3+4)$, $m_2(x)=(x^3+1)(x^3-1)(x^3+2)$, $m_3(x)=(x^3-1)(x^3+2)(x^3+4)$, $m_4(x)=(x^3+1)(x^3+2)(x^2-2)$, $m_5(x)=(x^3-1)(x^2-2)(x^3+4)$. Then, the lcm of all the moduli is $M(x)=(x^3+1)(x^3-1)(x^3+2)(x^2-2)(x^3+4)$. According to Proposition \ref{p1}, the code distance of the polynomial remainder code with the moduli is $d=3$. In addition, we can calculate $\tau_1=\tau_2=\tau_3=3, \tau_4=\tau_5=2$, and the error bound $\lambda<3$ in (\ref{gi}). Assume that there are one unrestricted error and an arbitrary number of bounded errors with degrees less than $3$ affecting the residue vector of a polynomial $a(x)$ with $\mbox{deg}\left(a(x)\right)<\mbox{deg}\left(M(x)\right)=14$.
 \begin{enumerate}
   \item For every $i$ with $1\leq i\leq 3$, take $\tilde{a}_i(x)$ as a reference and follow \textit{Algorithm \uppercase\expandafter{\romannumeral1}}. We want to calculate $\hat{k}_1(x),\hat{k}_2(x),\hat{k}_3(x)$, respectively. Note that some $\hat{k}_i(x)$ may not be reconstructed.
   \item Reconstruct $a(x)$ as $\hat{a}^{[i]}(x)=\hat{k}_i(x)m_i(x)+\tilde{a}_i(x)$ for each obtained $\hat{k}_i(x)$.
   \item Among  $\hat{a}^{[i]}(x)$ for $1\leq i\leq 3$, we can find at least two reconstructions $\hat{a}^{[i_1]}(x)$ and $\hat{a}^{[i_2]}(x)$ such that $\mbox{deg}\left(\hat{a}^{[i_1]}(x)-\hat{a}^{[i_2]}(x)\right)\leq\lambda<3$.
 \end{enumerate}
 Then, $a(x)$ is robustly reconstructed as $\hat{a}^{[i_1]}(x)$ or $\hat{a}^{[i_2]}(x)$.
\end{example}

\begin{remark}
The robust reconstruction in Section \ref{sec2} or Section \ref{sec4}, although, may not correct all the residue errors, if there is an additional error correction code on the top of it, it may be possible to correct all these residue errors, since only bounded errors are left in the reconstructions due to the robustness. As an example for the robust reconstruction in Theorem \ref{t2}, let the above $a(x)$ be encoded by a product code in polynomial residue number system, introduced in \cite{prodd}, i.e., given a polynomial $G(x)$, called the generator of the product code, a polynomial $a(x)$ with $\mbox{deg}(a(x))<\mbox{deg}(M(x))$ is legitimate in the product code of generator $G(x)$ if $a(x)\equiv0\mbox{ mod }G(x)$, else it is illegitimate. Using this product code on the top of robust reconstruction in Theorem \ref{t2}, $a(x)$ can be accurately determined as $a(x)=\hat{a}(x)-[\hat{a}(x)]_{G(x)}$ if $\mbox{deg}(G(x))>\lambda$, where $\lambda$ is the remainder error bound in (\ref{gi}) and $\hat{a}(x)$ is the robust reconstruction from Theorem \ref{t2}. This is due to the fact that $\mbox{deg}(a(x)-\hat{a}(x))\leq\lambda$ and thus $a(x)$ and $\hat{a}(x)$ have the same folding polynomial $(\hat{a}(x)-[\hat{a}(x)]_{G(x)})/G(x)$ with respect to the modulus $G(x)$. Furthermore, since $[a(x)]_{G(x)}=0$, i.e., $a(x)\equiv0\mbox{ mod }G(x)$, we have $a(x)=\frac{\hat{a}(x)-[\hat{a}(x)]_{G(x)}}{G(x)}\cdot G(x)=\hat{a}(x)-[\hat{a}(x)]_{G(x)}$.

While the above robust reconstruction has limitations in practice due to the type of bounded residue errors (i.e., only the last few coefficients of the polynomial residue are corrupted by errors), the theoretical result is new and may be interesting. In the next section, another motivation for us to study such bounded residue errors is shown for the improvements in uncorrected error probability and burst error correction in a data transmission. 
\end{remark}

\section{Correction of multiple unrestricted errors and multiple bounded errors in polynomial remainder codes}\label{sec3}

Let $m_1(x),m_2(x),\cdots,m_L(x)$ be $L$ non-pairwise coprime moduli, $M(x)$ be the lcm of the moduli, and $d$ be the code distance of the polynomial remainder code with moduli $m_i(x)$ for $1\leq i\leq L$. As one can see in the preceding section, a polynomial $a(x)$ satisfying $\mbox{deg}\left(a(x)\right)<\mbox{deg}\left(M(x)\right)$ can be robustly reconstructed when there are $t\leq\lfloor(d-1)/2\rfloor$ unrestricted errors and an arbitrary number of bounded errors with the remainder error bound $\lambda$ given by (\ref{gi}) in the residues $\tilde{a}_i(x)$ for $1\leq i\leq L$. Note that the reconstruction may not be accurate but robust and all the residues are allowed to have errors. Moreover, as stated in \cite{ss}, the polynomial remainder code with moduli $m_i(x)$ for $1\leq i\leq L$ can correct up to $\lfloor(d-1)/2\rfloor$ errors in the residues, i.e., $a(x)$ can be accurately reconstructed when there are only $t\leq\lfloor(d-1)/2\rfloor$ unrestricted errors in the residues, and any other residue is error-free. In this section, by making full use of the redundancy in moduli, we obtain a stronger residue error correction capability, that is, for the given set of moduli $m_i(x)$, in addition to $t\leq\lfloor(d-1)/2\rfloor$ unrestricted errors in the residues, some bounded residue errors can be corrected in the polynomial remainder code with moduli $m_i(x)$ and code distance $d$.

The remainder error bound $\eta(\theta)$ here, which depends on a variable $\theta$ for $1\leq \theta\leq L-2\lfloor(d-1)/2\rfloor$, is given by
\begin{equation}\label{eta}
\eta(\theta)<\tau_{(\theta)},
\end{equation}
where $\tau_{(\theta)}$ denotes the $\theta$-th smallest element in $\{\tau_1,\cdots,\tau_L\}$, and  $\tau_i=\min\limits_{j}\{\mbox{deg}\left(d_{ij}(x)\right),\mbox{ for }1\leq j\leq L,j\neq i\}$ for $1\leq i\leq L$. It is easy to see that the upper bound $\tau_{(\theta)}$ of the remainder error bound $\eta(\theta)$ increases as $\theta$ increases, i.e., the degrees of multiple bounded errors can be large as $\theta$ becomes large. Later, we will illustrate that the larger $\theta$ is, the smaller the number of correctable bounded errors is.

\begin{theorem}\label{theoremmm}
Let $m_i(x)$, $1\leq i\leq L$, be $L$ non-pairwise coprime polynomial moduli, $d$ denote the code distance of the polynomial remainder code with moduli $m_1(x),m_2(x),\cdots,m_L(x)$. Then, the polynomial remainder code can correct up to $\lfloor(d-1)/2\rfloor$ unrestricted errors and $\lfloor(L-\theta)/2\rfloor-\lfloor(d-1)/2\rfloor$ bounded errors as $\mbox{deg}\left(e_i(x)\right)\leq\eta(\theta)$ with the remainder error bound $\eta(\theta)$ given by (\ref{eta}) in the residues.
\end{theorem}
\begin{proof}
Without loss of generality, assume that $\tau_1\geq\tau_2\geq\cdots\geq\tau_L$. Similar to the proof of Theorem \ref{t2}, we take every residue in the first $L-\theta+1$ residues as a reference and want to calculate the corresponding folding polynomials by following \textit{Algorithm \uppercase\expandafter{\romannumeral1}}. After that, we reconstruct $a(x)$ as $\hat{a}^{[i]}(x)$ with each obtained folding polynomial as $\hat{a}^{[i]}(x)=\hat{k}_i(x)m_i(x)+\tilde{a}_i(x)$ for $1\leq i\leq L-\theta+1$. Since $\eta(\theta)<\tau_{(\theta)}\leq\tau_i$ for $1\leq i\leq L-\theta+1$, if a reference $\tilde{a}_i(x)$ is an error-free residue or a residue with a bounded error and the bound is $\eta(\theta)$, $\hat{k}_i(x)$ obtained from \textit{Algorithm \uppercase\expandafter{\romannumeral1}} is accurate according to Lemma \ref{31}. Since there are at most $\lfloor(d-1)/2\rfloor$ unrestricted errors and $\lfloor(L-\theta)/2\rfloor-\lfloor(d-1)/2\rfloor$ bounded errors in the residues, there are at most $\lfloor(L-\theta)/2\rfloor$ erroneous reconstructions in these $\hat{a}^{[i]}(x)$ for $1\leq i\leq L-\theta+1$, and the remaining reconstructions are correct and equal to each other. Therefore, we can find at least $\lceil(L-\theta)/2\rceil+1$ reconstructions $\hat{a}^{[i_j]}(x)$ for $1\leq j\leq\nu$ with $\nu\geq\lceil(L-\theta)/2\rceil+1$ such that
      \begin{equation}
      \hat{a}^{[i_1]}(x)=\hat{a}^{[i_2]}(x)=\cdots=\hat{a}^{[i_\nu]}(x).
      \end{equation}
Let $\hat{a}(x)$ be equal to the majority of all the $L-\theta+1$ reconstructions $\hat{a}^{[i]}(x)$ for $1\leq i\leq L-\theta+1$, i.e., $\hat{a}(x)=\hat{a}^{[i_j]}(x)$ for $1\leq j\leq\nu$. Then, $\hat{a}(x)$ is equal to the true $a(x)$, i.e., $\hat{a}(x)=a(x)$.
\end{proof}
\begin{remark}
From the above proof of Theorem \ref{theoremmm}, we now propose our new decoding algorithm for polynomial remainder codes with non-pairwise coprime moduli in the following. Without loss of generality, assume that $\tau_1\geq\tau_2\geq\cdots\geq\tau_L$.
\begin{enumerate}
  \item For every $i$ with $1\leq i\leq L-\theta+1$, take $\tilde{a}_i(x)$ as a reference and follow \textit{Algorithm \uppercase\expandafter{\romannumeral1}}. We want to calculate $\hat{k}_i(x)$ for $1\leq i\leq L-\theta+1$, respectively. Note that some $\hat{k}_i(x)$ may not be reconstructed.
  \item Reconstruct $a(x)$ as $\hat{a}^{[i]}(x)=\hat{k}_i(x)m_i(x)+\tilde{a}_i(x)$ for each obtained $\hat{k}_i(x)$. If the number of the obtained $\hat{k}_i(x)$ is less than $\lceil(L-\theta)/2\rceil+1$, the decoding algorithm fails.
  \item Among these reconstructions $\hat{a}^{[i]}(x)$ for $1\leq i\leq L-\theta+1$, if we can find at least $\lceil(L-\theta)/2\rceil+1$ reconstructions $\hat{a}^{[i_j]}(x)$ for $1\leq j\leq\nu$ with $\nu\geq\lceil(L-\theta)/2\rceil+1$ such that
      \begin{equation}
      \hat{a}^{[i_1]}(x)=\hat{a}^{[i_2]}(x)=\cdots=\hat{a}^{[i_\nu]}(x),
      \end{equation}
      let $\hat{a}(x)=\hat{a}^{[i_j]}(x)$ for $1\leq j\leq\nu$. Otherwise, the decoding algorithm fails.
\end{enumerate}

With the above decoding algorithm, if there are $\lfloor(d-1)/2\rfloor$ or fewer unrestricted errors and $\lfloor(L-\theta)/2\rfloor-\lfloor(d-1)/2\rfloor$ or fewer bounded errors with the remainder error bound $\eta(\theta)$ given by (\ref{eta}) in the residues, $a(x)$ can be accurately reconstructed from Theorem \ref{theoremmm}, i.e., $\hat{a}(x)=a(x)$. It is obviously seen that the price paid for the increased error correction capability is an increase in computational complexity. In the above decoding algorithm, \textit{Algorithm \uppercase\expandafter{\romannumeral1}} needs to be implemented $L-\theta+1$ times, i.e., the decoding algorithm in \cite{ss} (or in Section \ref{sec1} in this paper) used to reconstruct a folding polynomial in Step $4$ of \textit{Algorithm \uppercase\expandafter{\romannumeral1}} needs to be implemented $L-\theta+1$ times.
\end{remark}

\begin{example}\label{exa}
Let $m_1(x)=(x+1)(x+2)(x+3)(x+4)$, $m_2(x)=x(x+1)(x+3)(x+4)$, $m_3(x)=x(x+1)(x+2)(x+4)$, $m_4(x)=x(x+2)(x+3)(x+4)$, $m_5(x)=x(x+1)(x+2)(x+3)$ be $L=5$ moduli in $\mathrm{GF}(5)[x]$. We can easily obtain that the polynomial remainder code with the moduli has the code distance $d=4$, and $\tau_i=3$ for all $1\leq i\leq 5$. Therefore, from Theorem \ref{theoremmm} the polynomial remainder code can correct up to one unrestricted residue error and one bounded residue error with the remainder error bound $\eta(1)<3$. If one applies the result in \cite{ss}, only one unrestricted residue error can be corrected.
\end{example}

To see the improvements that are achieved in Theorem \ref{theoremmm}, we consider the application in a data transmission. In the residue number system, a number might be communicated from the sender to the receiver through the transmission of its residues. Instead of numbers, a method for transmitting information based on polynomials over a Galois field is used. To simplify the analysis, let a sequence be $a=(a[1],a[2],\cdots,a[k])$, where $a[i]\in \mathrm{GF}(p)$ for all $1\leq i\leq k$ and $p$ is a prime. Denote by $a(x)$ the corresponding polynomial $a(x)=\sum_{i=1}^{k}a[i]x^{i-1}$. Let moduli $m_1(x),m_2(x),\cdots,m_L(x)$ be $L$ polynomials in $\mathrm{GF}(p)[x]$ such that the degree of the lcm $M(x)$ of all the moduli is greater than $k-1$. If the degree of $m_i(x)$ is denoted by $m_i$ for each $1\leq i\leq L$, the corresponding residue $a_i(x)$ of $a(x)$ modulo $m_i(x)$ can be represented by $a_i(x)=\sum_{j=1}^{m_i}a_{ij}x^{j-1}$ for $a_{ij}\in\mathrm{GF}(p)$. In place of the original block $a$, the residue sequences $a_i=\left(a_{i1},a_{i2},\cdots,a_{im_i}\right)$ for $1\leq i\leq L$ are transmitted in the following order:
\begin{equation}\label{coe}
(a_{11},\cdots,a_{1m_1},a_{21},\cdots,a_{2m_2},\cdots,a_{L1},\cdots,a_{Lm_L}).
\end{equation}
If there is no error in the transmission, $a(x)$ can be accurately recovered using the CRT for polynomials in (\ref{crtpoly}), provided that $k$ and the moduli $m_i(x)$ are known. Then, $a$ is simply formed from the coefficients of $a(x)$. In practice, data can be corrupted during transmission. For a reliable communication, errors must be corrected. We herein consider two kinds of errors in the channel: random errors and burst errors.

Let the channel bit error probability be $\gamma$, the error probability of a residue $\tilde{a}_i$ be $p_{m_i}(\gamma)$, and the bounded residue error probability of $\tilde{a}_i$ be $q_{m_i}(\gamma;\theta)$, where $m_i$ is the length of the sequence presentation of moduli $m_i(x)$ over $\mathrm{GF}(p)$. Then,
\begin{equation}
p_{m_i}(\gamma)=1-(1-\gamma)^{m_i},
\end{equation}
\begin{equation}
q_{m_i}(\gamma;\theta)=(1-\gamma)^{m_i-\eta(\theta)}-(1-\gamma)^{m_i}.
\end{equation}
In what follows, let us consider the polynomial remainder code in Example \ref{exa}, where $d=4$ and $L=5$. We obtain two upper bounds for the uncorrected error probabilities in the decoding algorithm in Section \ref{sec1} and our proposed decoding algorithm, respectively.
\begin{itemize}
  \item According to Proposition \ref{p2}
  \begin{enumerate}
    \item The probability when all received residues are correct is
    \begin{equation}\label{eq1}
    p(c)=\prod_{i=1}^{L}(1-p_{m_i}(\gamma)).
    \end{equation}
    \item The probability when there is only one residue in error is
    \begin{equation}\label{eq2}
    p^\prime(c)=\sum_{i=1}^{L}p_{m_i}(\gamma)\cdot\underset{\begin{subarray}{ll}
j=1\\
j\neq i
\end{subarray}}{\mathop{\prod^{L}}}(1-p_{m_j}(\gamma)).
    \end{equation}
  \end{enumerate}
Then, from the decoding algorithm based on Proposition \ref{p2} in Section \ref{sec1}, we immediately obtain an upper bound for its uncorrected error probability as
  \begin{equation}
  P_{uncorrected}\leq 1-p(c)- p^\prime(c).
  \end{equation}

  \item According to Theorem \ref{theoremmm}
  \begin{enumerate}
    \item The probability when there are at most $\lfloor(L-\theta)/2\rfloor=\beta$ bounded errors in the residues is
     \begin{align}\label{eq3}
\begin{split}
\overline{p(c)}&=p(c)+\sum_{i=1}^{L}q_{m_i}(\gamma;\theta)\cdot\underset{\begin{subarray}{ll}
j=1\\
j\neq i
\end{subarray}}{\mathop{\prod^{L}}}(1-p_{m_j}(\gamma))\\
&\hspace{-1cm}+\sum_{i_1=1}^{L}\sum_{i_2>i_1}^{L}q_{m_{i_1}}(\gamma;\theta) q_{m_{i_2}}(\gamma;\theta)\cdot\underset{\begin{subarray}{c}
j=1\\
j\neq i_1,i_2
\end{subarray}}{\mathop{\prod^{L}}}(1-p_{m_j}(\gamma))\\
&\hspace{-1cm}+\cdots+\sum_{i_1=1}^{L}\sum_{i_2>i_1}^{L}\cdots\sum_{i_{\beta}>i_{\beta-1}}^{L}q_{m_{i_1}}(\gamma;\theta) \cdots q_{m_{i_\beta}}(\gamma;\theta)\\
&\hspace{-1cm}\cdot\underset{\begin{subarray}{c}
j=1\\
j\neq i_1,\cdots,i_\beta
\end{subarray}}{\mathop{\prod^{L}}}(1-p_{m_j}(\gamma)).
\end{split}					
\end{align}
    \item The probability when there are one error with degree greater than $\eta(\theta)$ and at most $\lfloor(L-\theta)/2\rfloor-1$ bounded errors with the remainder error bound $\eta(\theta)$ in the residues is
         \begin{align}\label{eq4}
\begin{split}
\hspace*{-0.28cm}\overline{p^\prime(c)}&=\sum_{i=1}^{L}\left(1-(1-\gamma)^{m_i-\eta(\theta)}\right)\cdot\left(\underset{\begin{subarray}{ll}
j=1\\
j\neq i
\end{subarray}}{\mathop{\prod^{L}}}(1-p_{m_j}(\gamma))\right.\\
&\left.\hspace{-1.3cm}+\underset{\begin{subarray}{ll}
i_1=1\\
i_1\neq i
\end{subarray}}{\mathop{\sum^{L}}}q_{m_{i_1}}(\gamma;\theta)\cdot\underset{\begin{subarray}{c}
j=1\\
j\neq i,i_1
\end{subarray}}{\mathop{\prod^{L}}}(1-p_{m_j}(\gamma))\right.\\
&\hspace{-1.3cm}\left.+\underset{\begin{subarray}{ll}
i_1=1\\
i_1\neq i
\end{subarray}}{\mathop{\sum^{L}}}\;\underset{\begin{subarray}{ll}
i_2>i_1\\
i_2\neq i
\end{subarray}}{\mathop{\sum^{L}}}q_{m_{i_1}}(\gamma;\theta)q_{m_{i_2}}(\gamma;\theta)\cdot\underset{\begin{subarray}{c}
j=1\\
j\neq i,i_1,i_2
\end{subarray}}{\mathop{\prod^{L}}}(1-p_{m_j}(\gamma))\right.\\
&\left.\hspace{-1.3cm}+\cdots+\underset{\begin{subarray}{ll}
i_1=1\\
i_1\neq i
\end{subarray}}{\mathop{\sum^{L}}}\;\underset{\begin{subarray}{ll}
i_2>i_1\\
i_2\neq i
\end{subarray}}{\mathop{\sum^{L}}}\cdots\underset{\begin{subarray}{c}
i_{\beta-1}>i_{\beta-2}\\
i_{\beta-1}\neq i
\end{subarray}}{\mathop{\sum^{L}}}q_{m_{i_1}}(\gamma;\theta)\cdots q_{m_{i_{\beta-1}}}(\gamma;\theta)\right.\\
&\left.\hspace{-1.3cm}\cdot\underset{\begin{subarray}{c}
j=1\\
j\neq i,i_1,\cdots,i_{\beta-1}
\end{subarray}}{\mathop{\prod^{L}}}(1-p_{m_j}(\gamma))\right).
\end{split}					
\end{align}
  \end{enumerate}
Then, from our decoding algorithm based on Theorem \ref{theoremmm}, we immediately obtain an upper bound for its uncorrected error probability as
 \begin{equation}
 \overline{P_{uncorrected}}\leq1-\overline{p(c)}- \overline{p^\prime(c)}.
 \end{equation}
\end{itemize}

It is obvious to see that $p(c)\leq\overline{p(c)}$ and $p^\prime(c)\leq\overline{p^\prime(c)}$. Therefore, we have $1-\overline{p(c)}- \overline{p^\prime(c)}\leq 1-p(c)- p^\prime(c)$. The performance of uncorrected error probabilities in Example \ref{exa} for the two decoding algorithms based on Proposition \ref{p2} and Theorem \ref{theoremmm} is shown in Fig. 1, where both simulations and the obtained upper bounds for uncorrected random errors are shown.

\begin{figure}[H]
  \hspace{-0.18in}
  \includegraphics[width=1.1\columnwidth,draft=false]{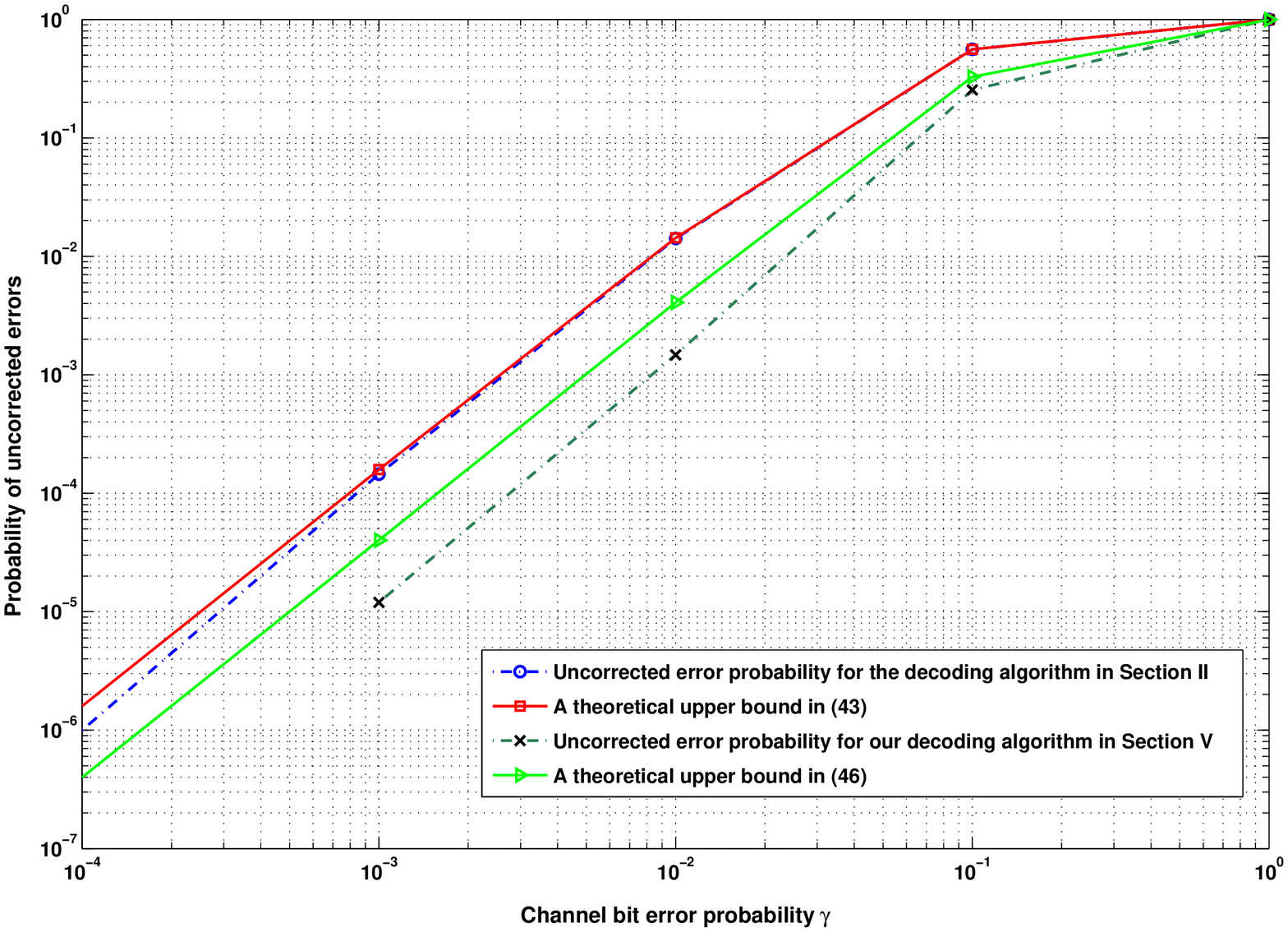}\\
  \vspace{-0.2cm}
  \begin{center}
Fig. 1. Uncorrected error probabilities based on Proposition \ref{p2} and Theorem \ref{theoremmm}: simulations and theoretical upper bounds.
\end{center}
\end{figure}

We next investigate the burst error correction capability in polynomial remainder codes with non-pairwise coprime moduli. As a residue $a_i(x)$ occupies $m_i$ bits, an error in this residue would affect up to $m_i$ bits. In order to express our question more precisely, we assume that all the moduli $m_1(x),m_2(x),\cdots,m_L(x)$ have the same degree $m$. Then, any error burst of width not more than $m+1$ in (\ref{coe}) can affect two residues at most. Similar to the result for polynomial remainder codes with pairwise coprime moduli in \cite{poly5,poly6}, it is directly obtained that the polynomial remainder code with non-pairwise coprime moduli and code distance $d$ that can correct up to $\lfloor(d-1)/2\rfloor$ errors in the residues can correct up to $\lfloor \lfloor(d-1)/2\rfloor/2\rfloor$ bursts of width not more than $m+1$, or correct one burst of width $(\lfloor(d-1)/2\rfloor-1)m+1$. By Theorem \ref{theoremmm}, however, we can further improve the capability of burst error correction in the polynomial remainder codes with non-pairwise coprime moduli. Let $A=\lfloor(d-1)/2\rfloor$ and $B=\lfloor(L-\theta)/2\rfloor-\lfloor(d-1)/2\rfloor$, and we have the following result.

\begin{figure}[H]
  \centering
  \includegraphics[width=1\columnwidth,draft=false]{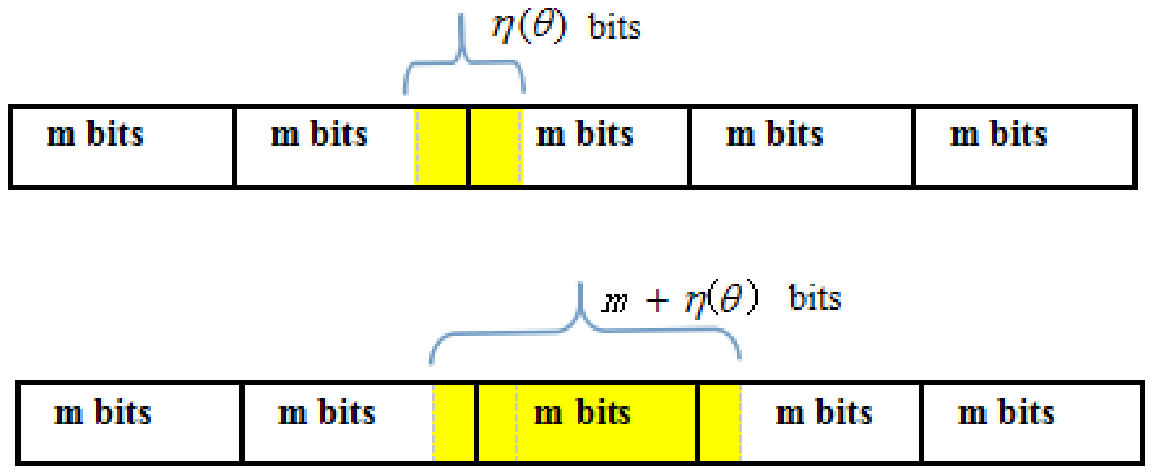}\\
   \vspace{-0.4cm}
  \begin{center}
Fig. 2. Burst error representations.
\end{center}
\end{figure}

\begin{corollary}\label{theorem8}
Let $m_i(x)$, $1\leq i\leq L$, be $L$ non-pairwise coprime polynomial moduli with the same degree $m$. Assume that the code distance of the polynomial remainder code with the moduli is $d$. Define by $M(x)$ with $\mbox{deg}\left(M(x)\right)>k-1$ the least common multiple of all the moduli, and let $\eta(\theta)$ be defined in (\ref{eta}). If a sequence $a=(a[1],a[2],\cdots,a[k])$ over $\mathrm{GF}(p)$ is encoded and sent by transmitting the coefficients of the residues of its corresponding polynomial $a(x)$ modulo $m_i(x)$ as in (\ref{coe}), then, based on residue error correction capability of polynomial remainder codes in Theorem \ref{theoremmm}, the following results are easily obtained:
\begin{itemize}
  \item[$1)$] Correction of bursts of width not more than $\eta(\theta)$
  \begin{enumerate}
    \item[$1.1)$] when $A\leq B$, it can correct up to $A$ such bursts;
    \item[$1.2)$] when $A>B$, it can correct up to $B+\lfloor(A-B)/2\rfloor$ such bursts.
  \end{enumerate}
  \item[$2)$] Correction of bursts of width not more than $m+\eta(\theta)$
  \begin{enumerate}
    \item[$2.1)$] when $\lfloor A/2\rfloor\leq B$, it can correct up to $\lfloor A/2\rfloor$ such bursts;
    \item[$2.2)$] when $\lfloor A/2\rfloor> B$, it can correct up to $B+\lfloor(A-2B)/3\rfloor$ such bursts.
  \end{enumerate}
\end{itemize}
\end{corollary}
\begin{proof}
From Theorem \ref{theoremmm}, the polynomial remainder code with moduli $m_i(x)$ for $1\leq i\leq L$ can correct up to $A$ unrestricted errors and $B$ bounded errors with the remainder error bound $\eta(\theta)$ in the residues. Fig. 2 shows that an error burst of width not more than $\eta(\theta)$ in (\ref{coe}) can, at most, give rise to a residue with an error and a residue with a bounded error. So, $1.1)$ and $1.2)$ are easily obtained. Similarly, from Fig. 2, an error burst of width not more than $m+\eta(\theta)$ in (\ref{coe}) can, at most, give rise to two residues with errors and one residue with a bounded error. So, when $\lfloor A/2\rfloor\leq B$, it can only correct up to $\lfloor A/2\rfloor$ such bursts. When $\lfloor A/2\rfloor> B$, in addition to correcting $B$ bursts of width not more than $m+\eta(\theta)$, the remaining error correction capability can correct up to $\lfloor(A-2B)/3\rfloor$ such bursts more.
\end{proof}
\begin{remark}
In the previous result in \cite{ss}, a polynomial remainder code with code distance $d$ can correct up to $A=\lfloor(d-1)/2\rfloor$ residue errors. Based on this error correction capability, it can only correct up to $\lfloor A/2\rfloor$ bursts of width not more than $\eta(\theta)$, or correct up to $\lfloor A/3\rfloor$ bursts of width not more than $m+\eta(\theta)$ if $\eta(\theta)>1$, which are not as good as the above results.
\end{remark}

\section{Conclusion}\label{sec5}
In this paper, we studied polynomial remainder codes with non-pairwise coprime moduli. We first considered the robust reconstruction problem from erroneous residues, namely robust CRT for polynomial problem, where all residues are allowed to have errors, but all the errors have to be bounded. A sufficient condition for the robustness bound was obtained, and a reconstruction algorithm was also proposed in the paper. Then, by releasing the constraint that all residue errors are bounded, another robust reconstruction was proposed when multiple unrestricted errors and an arbitrary number of bounded errors have occurred in the residues. Finally, compared with the previous residue error correction result in polynomial remainder codes, interestingly, our proposed result shows that in addition to correcting the number of residue errors as in \cite{ss}, some bounded residue errors can be corrected as well. With our proposed result in residue error correction, better performances in uncorrected error probability and burst error correction in a data transmission can be achieved.

\section*{Acknowledgment}
The authors would like to thank the Associate Editor and reviewers for their constructive comments that led to the improvement of this paper.


\begin{thebibliography}{1}

\bibitem{CRT1}
N. S. Szabo and R. I. Tanaka,
{\em Residue Arithmetic and its Application to Computer Technology},
 New York: McGraw-Hill, 1967.

\bibitem{CRT2}
H. K. Garg,
{\em Digital Signal Processing Algorithms: Number Theory, Convolution, Fast Fourier Transforms, and Applications},
Boca Raton, FL: CRC Press, 1998.

\bibitem{kls1}
H. Krishna, K. Y. Lin, and J. D. Sun, ``A coding theory approach to error control in redundant residue number systems. Part I: Theory and signal error correction,'' {\em IEEE Trans. Circuits Syst.}, vol. 39, pp. 8-17, Jan. 1992.

\bibitem{kls2}
J. D. Sun and H. Krishna, ``A coding theory approach to error control in redundant residue number systems. Part II: Multiple error detection and correction,'' {\em IEEE Trans. Circuits Syst.}, vol. 39, pp. 18-34, Jan. 1992.

\bibitem{vgoh}
V. T. Goh and M. U. Siddiqi, ``Multiple error detection and correction based on redundant residue number system,'' {\em IEEE Trans. Commun.}, vol. 56, pp. 325-330, Mar. 2008.

\bibitem{grs}
O. Goldreich, D. Ron, and M. Sudan, ``Chinese remaindering with errors,'' {\em IEEE Trans. Inf. Theory}, vol. 46, pp. 1330-1338, Jul. 2000.

\bibitem{rsk}
R. S. Katti, ``A new residue arithmetic error correction scheme,'' {\em IEEE Trans. Comput.}, vol. 45, pp. 13-19, Jan. 1996.

\bibitem{e121}
D. M. Mandelbaum, ``On a class of arithmetic codes and a decoding algorithm,'' {\em IEEE Trans. Inf. Theory}, vol. 22, pp. 85-88, Jan. 1976.

\bibitem{e122}
S. S. Yau and Y. C. Liu, ``Error correction in redundant residue number systems,'' {\em IEEE Trans. Comput.}, vol. 22, pp. 5-11, Jan. 1973.

\bibitem{mhe}
M. H. Etzel and W. K. Jenkins, ``Redundant residue number systems for error detection and correction in digital filters,'' {\em IEEE Trans. Acoust., Speech, Signal Process.}, vol. 28, pp. 538-545, Oct. 1980.

\bibitem{why}
Z. Gao, P. Reviriego, W. Pan, Z. Xu, M. Zhao, J. Wang, and J. A. Maestro, ``Efficient arithmetic-residue-based SEU-tolerant FIR filter design,'' {\em IEEE Trans. Circuits Syst. II, Exp. Briefs}, vol. 60, pp. 497-501, Aug. 2013.

\bibitem{eddc}
E. D. D. Claudio, G. Orlandi, and F. Piazza, ``A systolic redundant residue arithmetic error correction circuit,'' {\em IEEE Trans. Comput.}, vol. 42, pp. 427-432, Apr. 1993.

\bibitem{lly1}
A. S. Madhukumar and F. Chin, ``Enhanced architecture for residue number system-based CDMA for high-rate data transmission,'' {\em IEEE. Trans. Wireless Commun.}, vol. 3, pp. 1363-1368, Sep. 2004.

\bibitem{lly2}
L. L. Yang and L. Hanzo, ``Performance of a residue number system based parallel communications system using orthogonal signaling: Part I---System outline,'' {\em IEEE Trans. Veh. Technol.}, vol. 51, pp. 1528-1540, Nov. 2002.

\bibitem{lly3}
L. L. Yang and L. Hanzo, ``A residue number system based parallel communication scheme using orthogonal signaling: Part II---Multipath Fading channels,'' {\em IEEE Trans. Veh. Technol.}, vol. 51, pp. 1547-1559, Nov. 2002.

\bibitem{lle}
T. H. Liew, L. L. Yang, and L. Hanzo, ``Systematic redundant residue number system codes: Analytical upper bound and iterative decoding performance over AWGN and Rayleigh channels,'' {\em IEEE Trans. Commun.}, vol. 54, pp. 1006-1016, Jun. 2006.

\bibitem{lle2}
T. Keller, T. H. Liew, and L. Hanzo, ``Adaptive redundant residue number system coded multicarrier modulation,'' {\em IEEE J. Areas Commun.}, vol. 18, pp. 2292-2301, Nov. 2000.

\bibitem{papr}
Y. Yao, J. Hu, and S. Ma, ``A PAPR reduction scheme with residue number system for OFDM,'' {\em EURASIP J. Wireless Commun. Net.}, 2013, doi:10.1186/1687-1499-2013-156.

\bibitem{ass1}
A. Sengupta and B. Natarajan, ``Performance of systematic RRNS based space-time block codes with probability-aware adaptive demapping,'' {\em IEEE Trans. Wireless Commun.}, vol. 12, pp. 2458-2469, May 2013.

\bibitem{dad3}
S. Chessa and P. Maestrini, ``Dependable and secure data storage and retrival in mobile, wireless networks,'' in {\em Proc. Int. Conf. Dependable Syst. Netw.}, pp. 207-216, 2003.

\bibitem{dad1}
G. Campobello, S. Serrano, L. Galluccio, and S. Palazzo, ``Applying the Chinese Remainder Theorem to data aggregation in wireless sensor networks,'' {\em IEEE Commun. Lett.}, vol. 17, pp. 1000-1003, May 2013.

\bibitem{dad4}
G. Campobello, A. Leonardi, and S. Palazzo, ``Improving energy saving and reliability in wireless sensor networks using a simple CRT-based packet-forwarding solution,'' {\em IEEE/ACM Trans. Netw.}, vol. 20, pp. 191-205, Feb. 2012.

\bibitem{dad2}
M. Villari, A. Celesti, M. Fazio, and A. Puliafito, ``Evaluating a file fragmentation system for multi-provider cloud storage,'' {\em Scalable Computing: Practice and Experience}, vol. 14, pp. 265-277, Dec. 2013.

\bibitem{bm}
F. Barsi and P. Maestrini, ``Error codes constructed in residue number systems with non-pairwise-prime moduli,'' {\em Inf. Control}, vol. 46, pp. 16-25, Jul. 1980.

\bibitem{bm1}
R. S. Katti, ``A new residue arithmetic error corrction scheme,'' {\em IEEE Trans. Comput.}, vol. 45, pp. 13-19, Jan. 1996.

\bibitem{bm2}
A. Sweidan and A. A. Hiasat, ``On the theory of error control based on moduli with common factors,'' {\em Reliable Comput.}, vol. 7, pp. 209-218, 2001.

\bibitem{poly4}
I. S. Reed and G. Solomon, ``Polynomial codes over certain finite fields,'' {\em J. SIAM}, vol. 8, pp. 300-304, Oct. 1962.

\bibitem{poly44}
R. C. Bose and D. K. Ray-Chaudhuri, ``On a class of error correcting binary group codes,'' {\em Inf. Control}, vol. 3, pp. 68-79, Mar. 1960.

\bibitem{poly5}
J. J. Stone, ``Multiple-burst error correction with the Chinese remainder theorem,'' {\em J. SIAM}, vol. 11, pp. 74-81, Mar. 1963.

\bibitem{poly6}
D. C. Bossen and S. S. Yau, ``Redundant residue polynomial codes,'' {\em Inf. Control}, vol. 13, pp. 597-618, 1968.

\bibitem{poly8}
A. Shiozaki, ``Decoding of redundant residue polynomial codes using Euclid's algorithm,'' {\em IEEE Trans. Inf. Theory}, vol. 34, pp. 1351-1354, Sep. 1988.

\bibitem{peb}
P. E. Beckmann and B. R. Musicus, ``Fast fault-tolerant digital convolution using a polynomial residue number system,'' {\em IEEE Trans. Signal Process.}, vol. 41, pp. 2300-2313, Jul. 1993.

\bibitem{ss}
S. Sundaram and C. N. Hadjicostis, ``Fault-tolerant convolution via Chinese remainder codes constructed from non-coprime moduli,'' {\em IEEE Trans. Signal Process.}, vol. 56, pp. 4244-4254, Sep. 2008.

\bibitem{poly1}
J. H. Yu and H. A. Loeliger, ``On irreducible polynomial remainder codes,'' in {\em IEEE Int. Symp. on Information Theory}, Saint Petersburg, Russia, 2011.

\bibitem{poly2}
J. H. Yu, ``On the joint error-and-erasure decoding for irreducible polynomial remainder codes,'' arXiv:1202.5413, Feb. 2012.

\bibitem{poly3}
J. H. Yu and H. A. Loeliger, ``On polynomial remainder codes,'' arXiv:1201.1812, Jan. 2012.

\bibitem{chessa}
S. Chessa and P. Maestrini, ``Robust distributed storage of residue encoded data,'' {\em IEEE Trans. Inf. Theory}, vol. 58, pp. 7280-7294, Dec. 2012.

\bibitem{wenjie}
W. J. Wang and X.-G. Xia, ``A closed-form robust Chinese remainder theorem and its performance analysis,'' {\em IEEE Trans. Signal Process.}, vol. 58, pp. 5655-5666, Nov. 2010.

\bibitem{xiao}
L. Xiao, X.-G. Xia, and W. J. Wang, ``Multi-stage robust Chinese remainder theorem,'' {\em IEEE Trans. Signal Process.}, vol. 62, pp. 4772-4785, Sep. 2014.

\bibitem{prodd}
F. Barsi and P. Maestrini, ``Error detection and correction by product codes in residue number systems,'' {\em IEEE Trans. Comput.}, vol. c-23, pp. 915-924, Sep. 1974.


\end{thebibliography}
\end{document}